\documentclass[11pt]{article}

\usepackage{subfigure}
\usepackage{graphicx}
\usepackage[cmex10]{amsmath}
\usepackage{amsfonts}
\usepackage{amssymb}
\usepackage{booktabs}

\newcommand{\uno}{\>}
\newcommand{\due}{\>\>}
\newcommand{\3}{\>\>\>}

\newcommand{\var}[1]{{\it #1\/}}
\newcommand{\com}[1]{{\bf #1\/}}

\newtheorem{theorem}{Theorem}

\newtheorem{algorithm} {Algorithm}

\newtheorem{definition} {Definition}
\newtheorem{example} {Example}

\newtheorem{proposition} {Proposition}

\newenvironment{proof}[1][Proof]{\textbf{#1.} }{\ \rule{0.5em}{0.5em}}

\begin{document}

\title{On the Error Resilience of Ordered Binary Decision Diagrams}

\author{Anna Bernasconi\thanks{Dipartimento di Informatica, Universit\`a  di Pisa, Italy. {\tt anna.bernasconi@unipi.it}} \and Valentina Ciriani\thanks{Dipartimento di Informatica, Universit\`a degli Studi di Milano, Italy. {\tt valentina.ciriani@unimi.it}} \and Lorenzo Lago\thanks{Dipartimento di Informatica, Universit\`a degli Studi di Milano, Italy. {\tt lorenzo.lago@studenti.unimi.it}}}

\maketitle

\begin{abstract}

An Ordered Binary Decision Diagram (OBDD) is a data structure that is used in an increasing number of fields of Computer Science (e.g., logic synthesis, program verification, data mining, bioinformatics, and data protection)  for representing and manipulating discrete structures and Boolean functions.  
The purpose of this paper is to study the error resilience of OBDDs and to 
 design a resilient version of this data structure, i.e., a self-repairing OBDD. 
In particular, we describe some strategies that make reduced ordered OBDDs resilient to errors in the  indices, that are associated to the input variables, or in the pointers (i.e., OBDD edges) of the nodes.
These strategies exploit the inherent redundancy of the data structure, as well as the redundancy introduced by its efficient implementations. The solutions we propose allow the exact restoring of the original OBDD and are suitable to be applied to classical software packages for the manipulation of  OBDDs currently in use. 
Another result of the paper is the definition of a new canonical OBDD model, called {\em Index-Resilient Reduced OBDD}, which guarantees  that a node with a faulty index has a  reconstruction cost $O(r)$, where $r$ is the number of nodes with corrupted index. Experimental results on a classical benchmark suite validate the proposed approaches.
\end{abstract}

\section{Introduction}
\label{intro}

Ordered Binary Decision Diagrams (OBDDs) are the state-of-the-art data structure for Boolean function representation and manipulation. Indeed, they are widely used in logic synthesis, CAD of integrated circuits and in many safety critical applications, like verification  (see~\cite{BRB90} and~\cite{B92}, and~\cite{BCDV08,BCTV09,BCLP06} for more recent applications of OBDDs to logic synthesis). 
A binary decision diagram (BDD) over a set of Boolean variables $X=\{x_0, x_1,\ldots x_{n-1}\}$ is a rooted, connected direct acyclic graph, where each non-terminal node is labeled with a variable of $X$, and each terminal node is labeled with a value in $\{0,1\}$. Each non-terminal node has exactly two outgoing edges, $0$-edge and $1$-edge, pointing to two nodes  called $0$-child and $1$-child of the node. 
A BDD is \emph{ordered}  if there exists a total order $<$ over the set $X$ of variables such that if a non-terminal node is labeled by $x_{i}$, and its $0$-child and $1$-child have labels $x_{i_{0}}$ and $x_{i_{1}}$, respectively, then $x_{i}<x_{i_{0}}$ and $x_{i}<x_{i_{1}}$.

BDDs were first introduced by Lee~\cite{L59} and Akers~\cite{A78}, and developed by Bryant who proposed a canonical representation  in~\cite{B86}. 
Besides digital-system design, nowadays BDDs are  applied for representing and manipulating discrete structures 
in other research fields, as for instance data mining~\cite{M93,Minato10,Minato13}, bioinformatics~\cite{MI07,RC12,YNBD05}, and data protection~\cite{CDFLS12}.
The growing interest in BDDs is also evidenced by the fact that
in 2009 Knuth dedicated the first fascicle in the volume 4 of ``The Art of Computer Programming'' to this data structure~\cite{K09}. 

However, despite their popularity, error resilient versions of BDDs have not yet been proposed. 
We are aware only of a paper where security aspects of implementation techniques of OBDDs are discussed, and methods to verify the integrity of OBDDs are presented~\cite{D98}. 
In particular, a recursive checksum technique for on-line and off-line checks is proposed and experimentally evaluated: the on-line check verifies the correctness of the node during each access, so that errors can be detected very early; while the off-line check (usually performed by a depth-first-search algorithm starting from the rood of the OBDD) is  used to verify the integrity of the whole data structure.
However,~\cite{D98} only deals with the problem of error detection, and does not consider error correction, which is instead the main goal of our paper.

Nowadays, the resilience of algorithms and data structures to memory fault is a very important issue~\cite{FGI07,FGI09,I10}:  fast, large, and cheap memories in today's computer platforms are characterized by non-negligible error rates, which cannot be underestimated as the memory size becomes larger~\cite{JNW08}.  Computing in the presence of memory errors is therefore a fundamental task in many applications running on large, fast and cheap memories, as the correctness of the underlying algorithms may be jeopardized by even very few memory faults.

The scientific community has studied the problem in two different frameworks: {\em (i)} fault tolerant hardware design and {\em (ii)} development of error resilient algorithms and data structures. 
While fault tolerant hardware has been widely studied even in the past, the design of algorithms and data structures resilient to memory faults, i.e.,  algorithms and data structures that are able to perform the
tasks they were designed for, even in the presence of unreliable or corrupted information, has become much more attractive only recently (for a survey on the subject refer to~\cite{I10}).

The purpose of this paper is precisely to discuss the error resilience of OBDDs and to design a resilient version of this data structure. 

In particular, we describe some strategies that make reduced  OBDDs resilient to errors in the  indices, that are associated to the input variables, or in the pointers (i.e., OBDD edges) to the nodes.
These strategies exploit the inherent redundancy of this data structure, as well as the redundancy introduced by its efficient implementations. The solutions we propose {\em (i)} allow the exact restoring of the original OBDD and of the associated function $f$, and {\em (ii)} are suitable to be applied to classical software packages for the manipulation of  OBDDs currently in use, as for instance the CUDD library. 
Indeed, our first goal is to be able to efficiently reconstruct via software 
the corrupted OBDD  without changing the data structure.

However, to reach this goal we first assume that the unique table, i.e., a hash table used by most software implementation of OBDDs to facilitate their reduction (see  Section~\ref{prel} for more details),
 is  fault free. More precisely, we assume that the unique table is either implemented using error resilient linked lists~\cite{A96} or  it is stored in a safe memory area not affected by errors. The last one could be seen as a strong requirement, but fortunately we are able to remove this assumption completely still guaranteeing a very efficient reconstruction of all corrupted indices in the OBDD. 
Indeed,  the main contribution of the paper is the definition of a new canonical OBDD model, called {\em Index-Resilient Reduced OBDD}, which guarantees, by construction,  that a node with a faulty index has a  reconstruction cost $O(r)$, where $r$ is the number of nodes with corrupted index. As the new model does not exploit the unique table to restore all corrupted indices, we do not need a fault-free unique table anymore. Instead, we will only require that  the two terminal nodes (the leaves of the OBDD) are always uncorrupted, and therefore that they are memorized in a safe memory or duplicated. 
We also show how index-resilient reduced OBDDs can be constructed starting from binary decision trees or by applying Boolean operations  to  index-resilient reduced OBDDs. Both construction methods can be implemented with error resilient algorithms, i.e., algorithms capable of dealing with errors (in the data structures) occurring during their execution.

Finally, we describe some methods for dealing with errors on edges. We can consider two possible strategies: we use   safe unique tables, implemented with perfect hash functions, or we can use hash tables with error resilient linked lists~\cite{A96}. While the first approach guarantees a full error correction at the expense of the strong assumption on the fault freeness of the unique tables, the second strategy does not require safe unique tables, but can fail in some error corrections due to collisions. The experimental results indicate some setting for the hash tables that can limit the percentage of failed recoveries.

The paper is an extended version of the conference paper~\cite{BCL13} and is organized as follows.  Definitions  of the error models and preliminaries on OBDDs and their implementations are  described in Section~\ref{prel}. In Section~\ref{indices} we propose an efficient index reconstruction algorithm, and in Section~\ref{IR}  we introduce and study index-resilient OBDDs. 
Section~\ref{oper} discusses how   index-resilient OBDDs can be dynamically computed through a sequence of binary Boolean operators (as AND, OR, EXOR) applied to other index-resilient OBDDs  using the standard algorithm {\sc Apply} (reviewed in the Appendix). 
Section~\ref{edges} describes strategies for  broken edge  reconstruction. 
Experimental results for validating the proposed strategies  are reported in Section~\ref{exp}.
Section~\ref{concl} concludes the paper.

\section{Preliminaries}
\label{prel}

\subsection{Reduced Ordered Binary Decision Diagrams}
\label{prelBDD}

\begin{figure*}[ptb]
\begin{center}
\includegraphics[width=\textwidth]{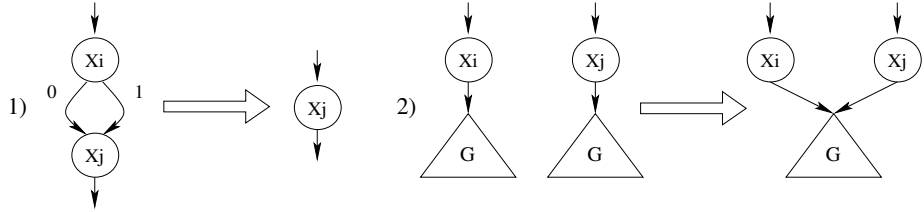}
\end{center}
\caption{\small \em Reduction rules for BDDs.}%
\label{figes7}%
\end{figure*}

A \emph{Binary Decision Diagram} (BDD) over a set of Boolean variables $X=\{x_0, x_1,\ldots x_{n-1}\}$ is a  rooted, connected direct acyclic graph, where each non-terminal (internal) node $N$ is labeled by a Boolean variable $x_i$ and has exactly two outgoing edges, the $0$-edge and the $1$-edge, pointing to two nodes called the $0$-child and the $1$-child of node $N$, respectively. $N$ is called the parent of its $0$- and $1$-child. Terminal nodes (leaves) are labeled $0$ or $1$.  For example, consider the BDD in Figure~\ref{fig:exampleBDD} with variables $x_0, x_1, x_2, x_3, x_4$. Each pointer to a $1$-child is depicted with a solid line, while each pointer to a $0$-child is depicted with a dashed line. 

Binary decision diagrams are typically used to represent Boolean functions. Let $f$ be a completely specified Boolean function, and $f_{x_{i}}$ and $f_{\overline{x}_{i}}$ be the functions resulting from $f$ when $x_{i}$ is $1$ and $0$, respectively.  The Shannon decomposition of $f$ around $x_{i}$ is:
$$
f=(x_{i}\land f_{x_{i}}) \lor (\overline{x}_{i} \land f_{\overline{x}_{i}})\,,
$$
where $\overline{x}_{i}$ is the negation of the variable $x_i$.
Any node in a BDD represents a Boolean function. The leaves represent the constant functions $0$ and $1$ and the root represents the entire Boolean function $f$. If the non-terminal node $N$ (with label $x_i$) represents the function $g$, then the $1$-child of $N$ (resp. $0$-child) represents the function $g_{x_{i}}$ (resp., $g_{\overline{x}_{i}}$).

The value of $f$ on the input $x_{0},\ldots,x_{n-1}$ is found by following the path indicated in the BDD by the values of $x_{0},\ldots,x_{n-1}$. A \emph{1-path} (resp. \emph{0-path}) in a BDD is a path from the root to a leaf labeled by $1$ (resp. $0$).
For example, consider the BDD in Figure~\ref{fig:exampleBDD}. 
The path that, starting from the root labeled with $a$, and corresponding to the variable $x_0$, goes through the nodes $b$, $d$, and $e$ (corresponding to $x_1$, $x_2$, and $x_3$) and arrives in the terminal $0$ is a 0-path. This path represents two possible inputs for the function $f$: $(0,0,0,0,0)$ and $(0,0,0,0,1)$, both with value 0 in $f$. 

\begin{figure}[t]
\centering
\subfigure[Binary Decision Diagram\label{fig:exampleBDD}]{\includegraphics[width=0.35\textwidth]{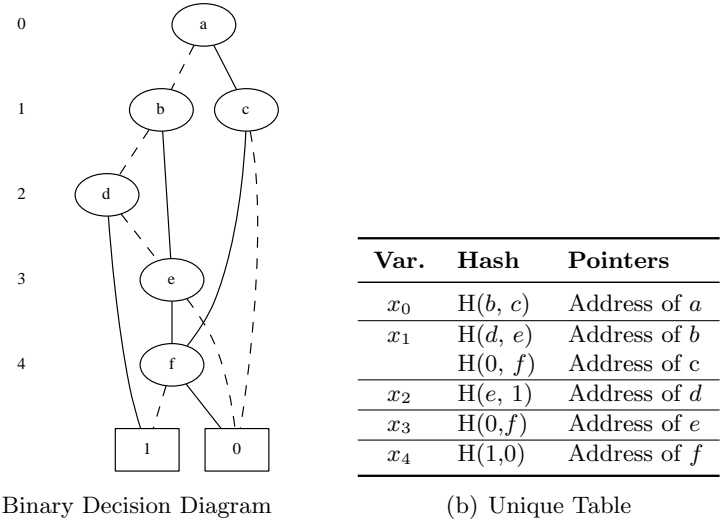}}
\hspace{5mm}
\subfigure[Unique Table\label{fig:exampleHash}]{
\begin{footnotesize}
\begin{tabular}[b]{cll} 
    \toprule
    \textbf{Var.} & \textbf{Hash} & \textbf{Pointers} \\
    \midrule
    $x_0$ & H($b$, $c$) & Address of $a$ \\
\hline
    {$x_1$} & H($d$, $e$) & Address of $b$ \\
          & H($0$, $f$) & Address of c \\
\hline
    {$x_2$} & H($e$, 1) & Address of $d$ \\
\hline
    {$x_3$} & H(0,$f$) & Address of $e$ \\
\hline
    {$x_4$} & H(1,0) & Address of $f$ \\
    \bottomrule
    \end{tabular}
\end{footnotesize}
}
\caption{A ROBDD and the corresponding unique table containing the subtables for the variables  $x_0, x_1, x_2, x_3$ and $x_4$.}
\end{figure}

In a BDD, each non terminal node $N$ is represented by the triple $$[N.\mbox{index},N.\mbox{0-child},N.\mbox{1-child}]$$ such that $N.\mbox{index}$ is the index $i$ of the variable $x_i$ which is the label of the node $N$, and $N$.0-child and $N$.1-child are the pointers to the $0$-child  and to the $1$-child of $N$, respectively.

A BDD is \emph{ordered} if there exists a total order $<$ over the set $X$ of variables such that if an internal node is labeled by $x_{i}$, and its $0$-child and $1$-child have labels $x_{i_{0}}$ and $x_{i_{1}}$, respectively, then $x_{i}<x_{i_{0}}$ and $x_{i}<x_{i_{1}}$. A BDD is \emph{reduced} if there exist no nodes whose $1$-child is equal to the $0$-child and there not exist two distinct nodes that are roots of  isomorphic subgraphs. A reduced and ordered BDD is called \emph{ROBDD}.
Starting from any OBDD we can obtain an equivalent reduced OBDD by repeatedly applying the following two rules: 
the \emph{Merge Rule} and the \emph{Deletion Rule} (see Figure~\ref{figes7}). According to the merge rule, if two nodes $M$ and $N$ have the same index, and their edges lead to the same nodes, then $N$ is deleted, and all the incoming edges of $N$ are redirected to $M$. Nodes $N$ and $M$ are called \emph{mergeable}. The deletion rule is used to remove {\em redundant nodes}, i.e.,  each  node $N$ that has both edges pointing to the same node $M$. In this case $N$ must be deleted and all its incoming edges redirected to $M$. When neither the merge rule nor the deletion rule can be applied, the OBDD is reduced.

The ROBDD is a canonical form; indeed, given a function $f:\{0,1\}^{n}\rightarrow\{0,1\}$ and a variable ordering $<$, there is exactly one ROBDD with variable ordering $<$ that represents $f$.

For example, consider the BDD in Figure~\ref{fig:exampleBDD}. The variable ordering is  $x_{0}<x_{1}<\ldots<x_{4}$, and each variable is represented by the corresponding level in the figure (i.e., the nodes at level $i$, with $0\leq i\leq 4$, are labeled with the variable $x_i$, i.e., have index $i$). For instance, the root node $a$ corresponds to the triple $[a.\mbox{index},a.\mbox{0-child},a.\mbox{1-child}] = [0, ${\em Address of node b}$, ${\em Address of node c}$]$. This BDD is ordered since any path from the root to a terminal node ($0$ or $1$) respects the variable ordering  $x_{0}<x_{1}<\ldots<x_{4}$. Moreover, the OBDD is reduced since we cannot apply the \emph{Merge Rule} or the \emph{Deletion Rule} to any node.

Many operations on Boolean functions can be efficiently implemented by OBDD's manipulations. For example Boolean operations (AND, OR, EXOR, etc.) between two OBDDs $g_1$ and $g_2$ have complexity $O(|g_1| \cdot |g_2|)$. 
The {\it if-and-only-if} operator ($\Leftrightarrow $), which tests two OBDDs for functional equivalence, has the same complexity.
The restriction of a function $f$, represented in an OBDD $B$ ($f_{x_i}$ or $f_{\overline{x}_i}$) can be computed in $O(|B|)$. Finally, the negation of a function $f$ has complexity $O(1)$.
For a description of these algorithms see~\cite{B86, B92} and the Appendix.

Note that the representation of Boolean functions with ROBDDs allows to perform operations that do not depend on the number of inputs that are equal to $1$ or $0$; for this reason, algorithms based on ROBDDs are usually defined implicit algorithms.  Usually, the terms BDD and OBDD are used instead of the correct term ROBDD.

There exists a wide variety of OBDD implementations.  Some of them focus mainly on operation efficiency, others on memory usage efficiency. However, some strategies are largely used and discussed in classical works.  One of them is the \emph{Unique Table}. The unique table ($\mathcal{U}$) is an array of hash tables (unique subtables), one for each variable of the function. We call $\mathcal{U}_{i}$ the unique subtable for the variable $x_i$. Each $\mathcal{U}_{i}$ contains the reference to all the nodes $N$ that contain variable $x_i$. This reference, usually the memory address of the node $[i,N.\mbox{0-child},N.\mbox{1-child}]$, is indexed using $N$.0-child and $N$.1-child as input of a hash function. The unique table is used to maintain the OBDD reduced: the lookup on the table is indeed used to determine whether it is necessary or not to create a new node. If a node $N$ with the same triple $[N.\mbox{index},N.\mbox{0-child},N.\mbox{1-child}]$  already exists in the OBDD, the lookup returns a pointer to that node; otherwise a new node is created. For example, Figure~\ref{fig:exampleHash} shows the unique table for the OBDD depicted in Figure~\ref{fig:exampleBDD}.

\subsection{Error Model}

The definition of a fault model is a key choice in designing resilient data structures. However, there is not a common fault model for data structures in literature~\cite{T90}: different authors have proposed different solutions, depending on the framework and on the type of errors considered.

A fault model in which any error is detectable via an error message when the program tries to reach the faulty object is proposed in~\cite{A96}. That work focuses on pointer-based data structures, as OBDDs. However, the authors assume that an error denies access to an entire node of the structure. Such granularity is not fine enough to catch some interesting cases. 
Consider for instance a node in a pointer-based data structure, as a list or a stack, usually composed of some data  and of one or more  pointers to other nodes. Certainly,  errors in different components of the node may affect differently the behavior of the data structure. The pointers, for instance, can maintain the structure properly connected despite an error in the data field of the node. Moreover, each component has peculiar characteristics, which could be exploited to increase the resilience of the data structure; these features might be lost if we consider only faults that involve the entire node.

A model with finer granularity, called {\em faulty-RAM}, is presented in~\cite{FGI07b,FI08,I10}. In faulty-RAM an adversary can corrupt any memory word and it is impossible to determine a priori if a memory area is corrupted or not. 
Such a scenario is realistic since an error can be induced by an external source, perhaps temporary, which can change any memory location that can not be discovered a priori. Consider for instance a minor change, e.g., a single bit, in a memory location storing an integer value:  the result is another integer, whose incorrect origin cannot always be detected.
Another  characteristic of the faulty-RAM model is its fine granularity: any memory location (from a single bit, the single data, or an entire structure) can be affected by a fault.

Another interesting error model is the {\em single-component model}~\cite{T90}, which focuses on single attributes of an item at a time and assumes that each error affects one component of one node of the storage structure, e.g., a pointer, a counter, an identifier field.
As mentioned earlier, reasoning at the component level allows us to exploit in a deeper way the characteristics of the nodes and of the whole data structure. For example, consider the basic representation of a node in a list: a fault can affect the given node, or the pointer to the next node. The two components have very different characteristics: if the loss of some data fields of a node, excluding the pointer,  can be tolerated  in certain conditions, the loss of the pointer can make unreachable a part of the data structure.

A further step forward is to model macro-faults such as {\em Copy Faults} (the content of a node is copied incorrectly into another node) or {\em Memory Allocation Faults} (memory is allocated that has already been used by another node), as proposed in~\cite{D98} for OBDD's integrity verification. 
This model can be seen as a particular case of faulty-RAM or single-component model where memory faults create situations hard to spot (for example a corrupted pointer points to another node instead to a meaningless memory area). 

In this paper we use the single-component model, and we consider, as components of a node, the index $i$ and both the $0$ and $1$-pointers. 
We assume perfect error detection capabilities: errors are immediately reported when the program tries to use the fault component of a node.
In fact, the main goal of our analysis is to study the capability of this data structure to restore corrupted data, not to detect them. However, it is worth mentioning that the peculiar structure of OBDDs could be exploited for error detection too. For instance, the presence of faulty indices can be reported any time the index of a node and those of its children in the diagram are not consistent with the fixed  variable ordering.

In our analysis, we also assume that the unique table of an OBDD is implemented using fault tolerant linked lists~\cite{A96}. 
We recall from~\cite{A96}  that  fault tolerant linked lists are resilient up to $d$ faults, where $d$ is a parameter, and present an $O(1)$ space and   amortized time overhead  with respect to the basic data structure. In presence of $f < d$ faults, at most $O(f \log f)$ nodes of the lists are lost. Moreover, each node has a constant size and a constant out-degree, and the reconstruction time is a small polynomial in $f$ and $d$, independently of the list size.

An alternative, but less practicable, assumption is to store the unique table of an OBDD in a safe memory area, which is not affected by errors.

Finally observe that, our analysis implicitly assumes  that an OBDD is constructed correctly, and that  memory faults occur when the data structure is in use. This could be seen as a strong assumption, but fortunately this assumption can be completely removed for index-resilient reduced OBDDs (see Section~\ref{IR} for more details). 

\section{Errors in Indices}
\label{indices}
As we have seen, the OBDD node core structure is made up of three elements: an index $i$ and the   pointers to the $0$-child and to the $1$-child of the node. On this simple structure two types of faults can occur: corruption of the index or corruption of a pointer. In this section we discuss error resilient indices in OBDDs. In particular, we propose an efficient reconstruction algorithm, analyze 
the cost  of the reconstruction of a corrupted index, and study the impact of the OBDD reduction rules on this cost. This study gives us the knowledge to describe in Section~\ref{IR} a new and efficient variant of OBDDs that is index resilient.

\subsection{Reconstruction Algorithm}

First of all, we show  how to reconstruct the index of a faulty node, restoring exactly the  original OBDD and the  associated function $f$.

Let $N$ be a node, described by the triple $[N.\mbox{index}, N.\mbox{0-child}, N.\mbox{1-child}]$, and suppose that a fault occurred on the index of $N$, so that $N$ cannot be associated with one of the input variables. This causes a problem as it is impossible to determine the value of the function $f$, represented by the OBDD, on all input assignments whose corresponding paths go through  $N$.

Without loss of generality, let us assume that the chosen variable ordering is $\{x_0, x_1, \ldots, x_{n-1}\}$, so that the index of a variable defines the level of the variable in the corresponding OBDD.
A first attempt to reconstruct the index of the faulty node $N$ is to define a range of indices that contains the original index of the node. 

\begin{definition}[Node range]
Let $N$ be a node in an OBDD $B$, $I_N = [i_P+1, i_C-1]$  is the range containing all the possible levels for  $N$ in $B$, where $i_P$ is the maximum index of $N$'s parents in $B$, and $i_C$ is the minimum index of its children.
\end{definition}
If $i_P+1 = i_C-1$, then the lost index $i$ is $i_C - 1$ (or $i_P + 1$). Otherwise, we cannot say which  index in the range was the original one.

\begin{example}
\label{ex: Intervallo}
Consider the OBDD in Figure~\ref{fig:exampleBDD}. Suppose that the index $1$ (corresponding to label $x_1$) of the node $c$ is faulty.  Note that the variable ordering of the OBDD is $x_0, x_1, x_2, x_3, x_4$. Since the faulty node has a parent node at level 0 (i.e., the parent has label $x_0$), the first possible level for $c$ is 1. Moreover, the minimum index child of $c$ has label $x_4$, which means that the index of node $c$ can be at most 3. In summary we have that $i_P=0$, $i_C=4$ and thus $I_c = [1, 3]$. 
\end{example}
Once we have defined the range $I_N$ containing the possible indices for $N$, we can use the unique table to find the correct index of $N$ in this range as depicted in Algorithm~\ref{fig:algoindex}.
For each possible value $l$ in the range $I_N$, the algorithm visits the collision list corresponding to the hash value Hash($N$.0-child, $N$.1-child)
in the unique subtable associated to index $l$; that is, the algorithm examines all nodes with index $l$ and pointers equals to the ones of the faulty node $N$,  until it finds  a node with the same memory address of $N$.

\begin{algorithm} [Reconstruction of the faulty index]
\label{fig:algoindex}
\ \ 

\begin{scriptsize}
\hrule 

\begin{tabbing}
{\bf INPUT}\\
$N$ /* Address of the faulty node  */\\
$BDD$ /* OBDD containing $N$ */\\
{\bf OUTPUT}\\
\var{Index} /* Correct index */\\ \\

\var{$I_N$} = ({\em MaxLevel}($N$.Parents), {\em MinLevel}(\{{$N$.0-child, $N$.1-child}\})) \\
\com{for} \= \com{each} $l\in I_N$ \com{do} \\
\uno		/* unique subtable for the index $l$ */ \\
\uno		uniqueTable \= = OBDD.UniqueTables[$l$] \\
\uno 		node\= = uniqueTable[Hash($N$.0-child, $N$.1-child)] \\
\uno			\com{while} \= (node $\not= N$ $\land$ node.Next $\not=$ NULL)\\
\due				/* visit of the collision list */ \\
\due				node = node.Next \\
\uno			\com{if}(node == N)	\com{return} $l$\\
		\com{return} -1

\end{tabbing}
\vspace{-8pt} 
\hrule
\end{scriptsize}
\end{algorithm}
Let us examine, through an example, how this algorithm works.
\begin{example}
Consider the OBDD in Figure~\ref{fig:exampleBDD}, together with its unique table of Figure~\ref{fig:exampleHash}. Suppose that $c$ is the faulty node. We know from Example~\ref{ex: Intervallo} that $I_c=[1,3]$.
Suppose to examine $I_c$ starting from level $l = 3$.  
The unique subtable associated to the variable with index $3$ contains a node with the same pointers  as the faulty node $c$, i.e., pointers to the terminal node $0$ and to node $f$, respectively, but the address of this node is different from that of $c$.
The algorithm then considers level $l = 2$: the unique subtable associated to the variable $x_2$ does not contain any node with the same pointers as $c$. Finally, for $l = 1$, a match is found: the unique subtable corresponding to the variable  $x_1$ contains a node with the same pointers as $c$ and the same memory address. Thus the correct index of the faulty node $c$ is 1.
\end{example}
In the next proposition, we prove that Algorithm~\ref{fig:algoindex} is correct: it always outputs a result, which is exactly the index of the faulty node (in the single error model).
\begin{proposition}
Let  $B$ be an ordered  OBDD, and $N$ the only node whose index $N.\mbox{index}$ is corrupted. Algorithm~\ref{fig:algoindex} always outputs an index $j$ such that $j=N.\mbox{index}$.
\end{proposition}
\begin{proof}
Recall that the  unique table is an array of hash tables (unique subtables), each corresponding to a variable of the function. The subtable  associated to the variable $x_i$ ($0 \le i < n$) contains  the reference to all the nodes $N$ labeled by $x_i$. This reference, usually the memory address of the node $[i,N.\mbox{0-child},N.\mbox{1-child}]$, is indexed using $N$.0-child and $N$.1-child as input of a hash function.

Assume, by contradiction, that the algorithm does not output any result. This means that the unique subtable  associated to the variable with index $N.\mbox{index}$ does not contain a pointer to node $N$ in the   cell (or list in case of collisions of the hash function) corresponding to the hash value  $Hash(N.\mbox{0-child}, N.\mbox{1-child})$. Thus we reach a contradiction as such a situation can arise only if  $N \not\in B$.

Now suppose that  the algorithm outputs an index $j$ such that $j \neq N.\mbox{index}$.
This means that the unique subtable associated to  $j$ contains, in the  cell (or collision list) corresponding to the hash value  $Hash(N.\mbox{0-child}, N.\mbox{1-child})$, 
 a pointer to a node $N' = [j, N.\mbox{0-child}, N.\mbox{1-child}]$  stored in the  same memory area of $N$. 
This is again a contradiction. In fact (i) a  pointer to a node is  always inserted  in the unique subtable corresponding to its index, thus the pointer to $N$ cannot be stored in the unique subtable of $j \neq N.\mbox{index}$; and (ii) there cannot exist a node $N' \neq N$ in the same memory area of $N$, as  different nodes cannot be stored in the same memory location.
\hfill\end{proof} 

We can note that the reconstruction of a faulty node $N$ costs $O(|I_N|)$ on average, since operations on the unique hash table have an average constant time complexity.

Observe that the algorithm can handle at most one error, thus it outputs an invalid index value  -1 only in presence of more faulty nodes, when it cannot work as expected. 
Such a situation could arise for instance if the index of a parent or of a child of the faulty node $N$ is corrupted as well. Indeed, in this case the range $I_N$ computed by the algorithm would not be correct as, e.g., it could not contain the level of the faulty node.
However, such a situation can be easily handled in the following way: if, after scanning $I_N$, the algorithm has not found  the index of $N$, the range $I_N$ can be extended considering  lower values for the parents' indices, and higher values for the index of the children of $N$.
In the worst case, when all indices are corrupted, we must set $I_N= [0,n]$ as range of the root of the OBDD, to reconstruct the index of the root, while for all other nodes in the OBDD we will be able to set only a lower limit for their range, if we restore all indices from the root down to the terminal nodes.
 
\subsection{Reconstruction Cost}
Let us now examine which characteristics make an OBDD more suitable to the reconstruction of a corrupted index, that is on which diagrams the proposed algorithm is more efficient. To this aim, we introduce a metric to measure the cost of the reconstruction of a corrupted index of an OBDD node in the worst case, the overall cost of index reconstruction for all nodes in an OBDD, and the average reconstruction cost.

\begin{definition}[Index reconstruction cost]
The reconstruction cost $C(N)$ of the faulty index $N$ is given by the number of indices that are candidate to be the correct one in $N$.
\end{definition} 
If we consider the case of one fault only in node $N$, we have that $C(N)$ is at most $|I_N|$. In particular, $C(N) = |I_N|$ whenever there is no additional knowledge on the structure of the OBDD. In the rest of this section, we therefore assume that $C(N) = |I_N|$. Instead, in Section~\ref{IR} we will study OBDDs with a particular structure implying that $C(N) \le |I_N|$.

For example, the reconstruction cost of the node $c$ of the OBDD in Figure~\ref{fig:exampleBDD} is the cardinality of its range, i.e., $C(c) = |I_c| = |[1,3]| = 3$.

\begin{definition}[Overall index reconstruction cost]
Given an OBDD $B$ with $k$ nodes  $\{N_1, N_2, ..., N_k\}$, the {\em overall index reconstruction cost} of its nodes is  $$C_t(B) = \sum_{N \in \{N_1, ..., N_k\}} C(N)\,.$$
\end{definition}

\begin{definition}[Average index reconstruction cost]
Given an OBDD $B$ with $k$ nodes, the {\em average index reconstruction cost} of its nodes is  $$C_m(B) = \frac{C_t(B)}{k}\,.$$ 
\end{definition}

In the best case, $C_m$ is a constant, meaning that the index of each node of the OBDD can be reconstructed in constant time.
This condition is satisfied, e.g., by a complete OBDD, where no reduction rules have been applied. In fact,  in a ``complete'' unreduced OBDD, all paths from the root to the terminal nodes contain exactly $n$ nodes, where $n$ is the number of input variables. Thus, for each node $N$, $C(N) = |I_N| = 1$.

It is interesting to notice that the optimal cost $C_m(B) = 1$ can also be reached by reduced OBDD, as it happens, e.g., for the parity function, whose OBDD, even  if  very compact, with at most two nodes per level, only contains paths of length $n$, i.e., path with a node on each level (see Fig.~\ref{fig:exampleParity}). 

\begin{figure}[t]
\subfigure[Parity function\label{fig:exampleParity}]{\includegraphics[width=0.40\textwidth ]{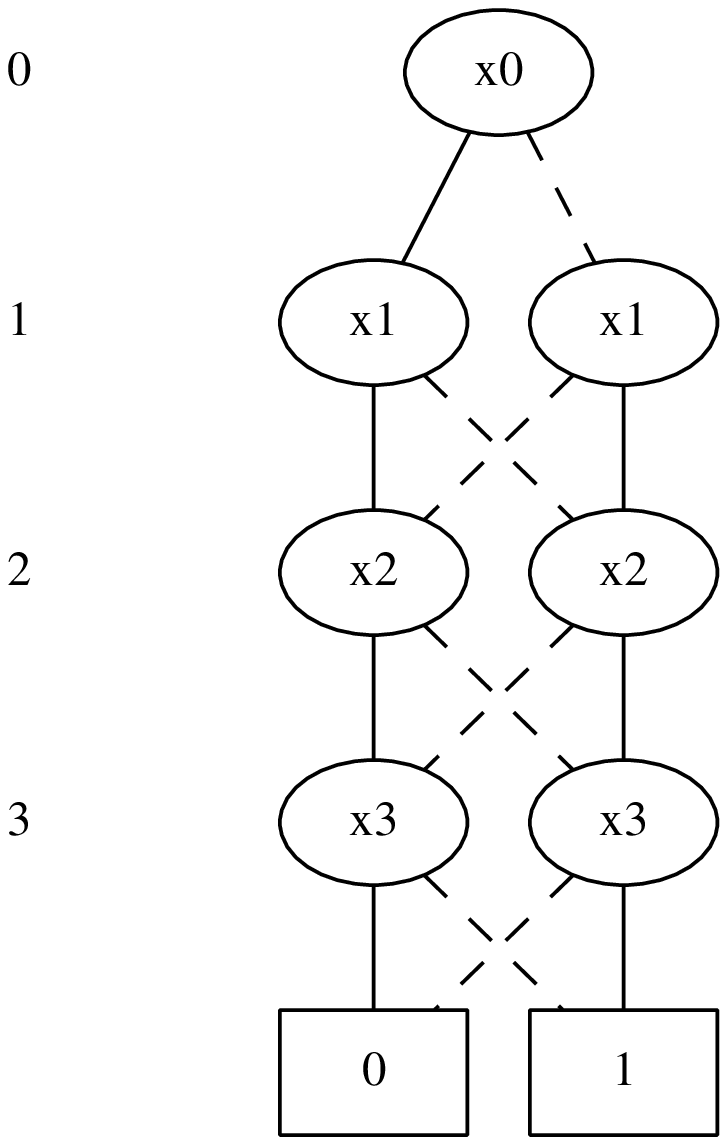}}
\subfigure[$\overline x_3 \land ((\overline x_0 \land \overline x_1) \lor \overline x_2)$\label{fig:exampleCm1}]{\includegraphics[width=0.40\textwidth ]{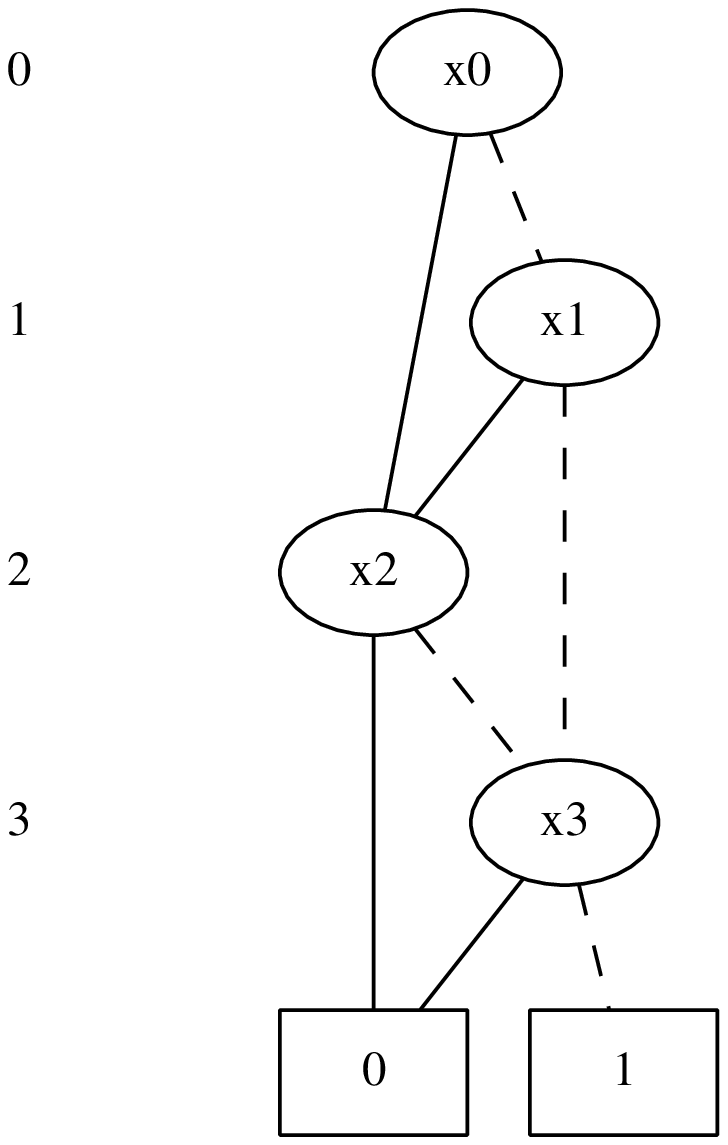}}
\caption{Examples of two ROBDDs with reconstruction cost  equals to 1.}
\label{fig:exampleROBDDCm1}
\end{figure}

\begin{figure}[t]
\begin{center}
\includegraphics[width=0.30\textwidth]{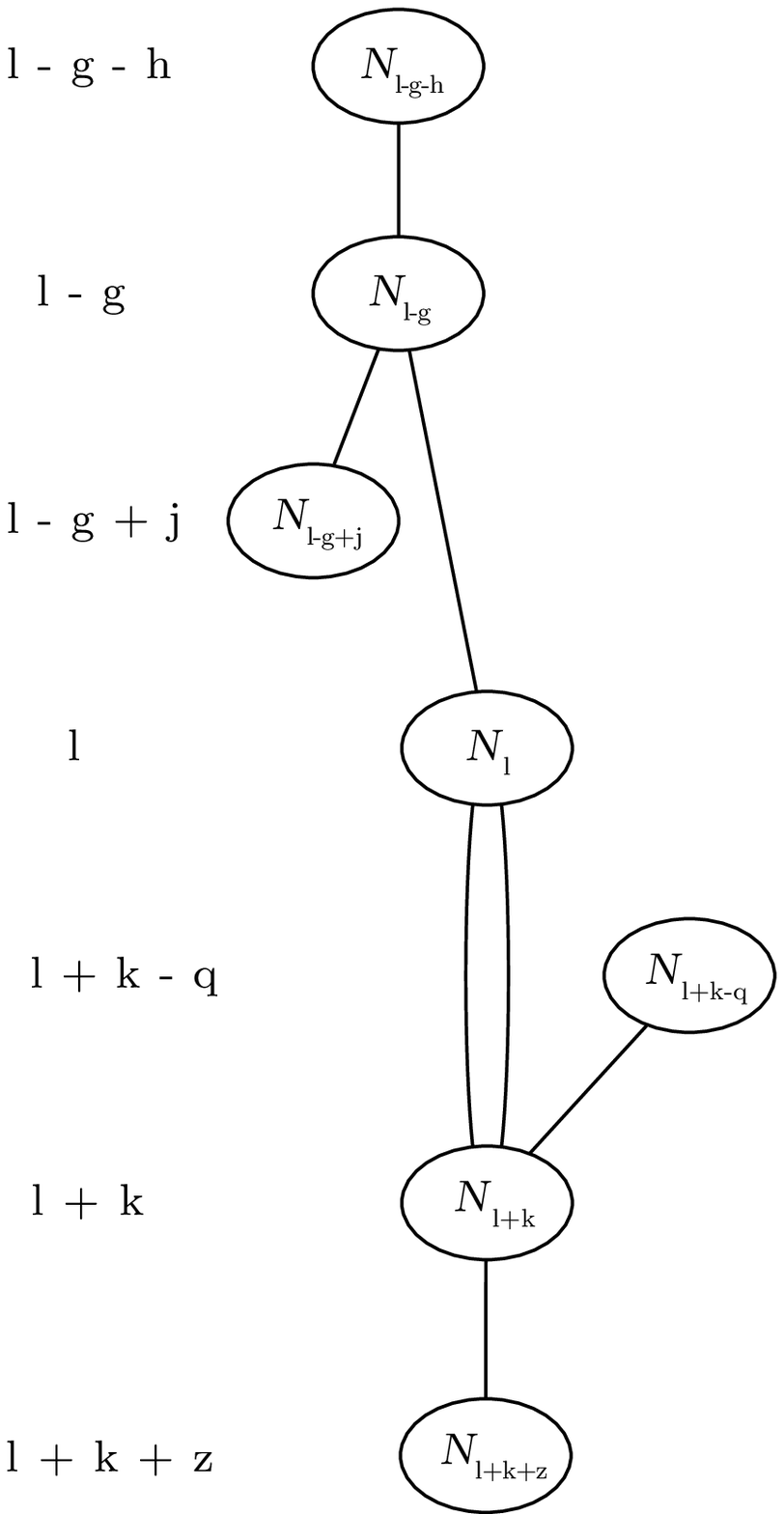}
\end{center}
\caption{Part of a ROBDD described in Theorem~\ref{the:reducereg1}.}
\label{fig:reducereg1}
\end{figure}

As these two examples (complete unreduced OBDD, and OBDD for the parity function) clearly suggest, the reconstruction cost increases whenever an OBDD contains paths, from the root to the terminal nodes,  shorter than $n$, i.e.,  paths lacking nodes from some level of the diagram. In fact, in this case, the range of the nodes possibly increases.
Such a condition is caused by the application of the reduction rules to the starting complete OBDD representing a given function.
Thus, let us examine the impact of such rules on the index reconstruction cost.
Let us start with the first rule, i.e.,  the merge of isomorphic subgraphs.
\begin{theorem}
\label{the:reducereg1}
Let $\{N_1, N_2, ..., N_m\}$ be the roots of the isomorphic subgraphs to be merged, and let  $N_i \in \{N_1, N_2, ..., N_m\}$ be the node that will be kept in the diagram after the merge. The application of this reduction rule improves the overall index reconstruction cost $C_t$ that decreases by 
a value equal to $\sum_{j=0, j \not= i}^{m} C(N_j)$.
\end{theorem}
\begin{proof}
First observe that the roots of the  isomorphic subgraphs  $\{N_1, N_2, ..., $ $ N_m\}$ have the same index value,  are all on the same level $l$, and share the same children. After the application of the rule, $m-1$ of these nodes are deleted from the OBDD, and the edges pointing to them are all redirected to the only root node that is kept. Thus, the ranges of the parents of the deleted nodes do not change, as they are redirected to a node on the same level and with the same index of their original child; and for the same reason, the ranges of the children of the deleted nodes do not change either. 
Finally, note that the application of the rule allows us to subtract from the overall cost $C_t$ the index reconstruction cost of all deleted nodes. 
\hfill\end{proof} 

Let us now examine the impact of the second rule, i.e., the deletion of nodes with both edges pointing to the same OBDD node.
Let $N_l$ be a node at level $l$ whose 0-edge and 1-edge point to the same node $N_{l+k}$ on level $l+k$, with $k > 0$. Let $N_{l+k+z}$ be the child of $N_{l+k}$ with minimum index  $l+k+z$ with $z>0$. Let  ${N_{l-g_1}, N_{l-g_2}, \ldots, N_{l-g_r}}$ be the $r$ parents of  $N_l$, with $N_{l-g_i}$ on level  $l-g_i$, with  $ 0<g_i \leq l$, for $1 \le i \le r$. Let $N_{l-g_i-h_i}$ be the parent of $N_{l-g_i}$ with maximum index $l-g_i-h_i$, with $ 0<h_i< l-g_i$. Let $N_{l-g_i+j_i}$ be the other child  of $N_{l-g_i}$ on level $l-g_i+j_i$, with $j_i>0$. Finally, let $N_{l+k-q}$ ($N_{l+k-q} \neq N_l$) be the parent of $N_{l+k}$ (if exists) with the highest index $l+k-q$ with $0 < q \leq l+k$.
Figure~\ref{fig:reducereg1} shows the portion of diagram described above.  
\begin{theorem}
\label{the:reducereg2}
After the application of the deletion rule, the overall index reconstruction cost of the OBDD changes for an amount $\delta \in [-\min(g_1,g_2,\ldots,g_r)$ $ - k -1, k (r -1)+1].$
\end{theorem}
\begin{proof}
Before the deletion of $N_l$, the ranges $I_{l-g_i}, I_l, I_{l+k}$ of the nodes $N_{l-g_i}, N_l, N_{l+k}$, respectively, are (see Figure~\ref{fig:reducereg1}): 
\begin{eqnarray*}
&\ &I_{l-g_i} = [l-g_i-h_i, \min(l, l-g_i+j_i)]\,,\\
&\ &I_l = [l-\min(g_1, g_2, \ldots, g_r), l+k]\,,\\
&\ &I_{l+k} = [\max(l,l+k-q), l+k+z]\,.
\end{eqnarray*}

After the deletion of $N_l$, these ranges change as follows (see Figure~\ref{fig:reducereg1}):
\begin{eqnarray*}
&\ &I_{l-g_i} = [l-g_i-h_i, \min(l+k, l-g_i+j_i)]\,,\\
&\ &I_{l+k} = [\max(l-\min(g_1, g_2, \ldots, g_r), l+k-q), l+k+z]\,.
\end{eqnarray*}
$I_l$ is empty since $N_l$ has been deleted. 
The lower bound for $I_{l-g_i}$ does not change, since the deletion happened at a lower level, while the upper bound now depends on the child of $N_l$. On the other hand, $I_{l+k}$ maintains its upper bound, but its lower bound now depends on the parent of $N_l$ with highest index. 
Note that the ranges of the other nodes do not depend on $l$, thus they do not change after the deletion of $N_l$.

The best case happens when no range depends on $N_l$, that is $\forall{i \in r}: l-g_i+j_i \leq l$ and $l \leq l+k-q$. In this case, the ranges $I_{l-g_i}$ and $I_{l+k}$ do not change. Since $N_l$ is removed, the overall index reconstruction cost $C$ is reduced by
$$|I_l| = \min(g_1,g_2,\ldots,g_r) + k +1.$$ 

The worst case happens when all the parents and the child of $N_l$ have a range depending on $N_l$, i.e., $\forall{i \in [1,\ldots, r]}: l-g_i+j_i > l$ and $l > l+k-q$. In this case, the deletion of $N_l$ implies a change of the ranges $I_{l-g_i}$ and $I_{l+k}$. The upper bound for each $N_{l-g_i}$ is now $\min(l-g_i+j_i, l+k)$. Thus, $|I_{l-g_i}|$ is increased by $\min (j_i-g_i,k) $ for each $1\leq i \leq r$. Therefore, the increase due to $|I_{l-g_i}|$ for all ${i \in [1,\ldots, r]}$, is 
$$\sum_{i=1}^r \min (j_i-g_i,k) \leq  \, \sum_{i=1}^r k = r\,  k.$$

 The increase of $|I_{l-g_i}|$ is then upper-bounded by $r \, k$. Moreover, the lower bound for
$N_{l+k}$ changes from $l$  to $\max(l-\min(g_1, g_2, \ldots, g_r), l+k-q)$. Thus, the increase of $|I_{l+k}|$ is
$$ l - \max(l-\min(g_1, g_2, \ldots, g_r), l+k-q) \leq \min(g_1, g_2, \ldots, g_r).$$
Finally, since $|I_l| =  \min(g_1,g_2,\ldots,g_r) + k -1$ and $N_l$ is deleted, the value of $C$, in the worst case, is increased by: 
$$ k \, r + \min(g_1, g_2, \ldots, g_r) - \min(g_1, g_2, \ldots, g_r) - k + 1 = k ( r -1)+1.$$
\hfill\end{proof}

Note that the use of the deletion rule does not always increase the index reconstruction cost. For example, consider the reduced OBDD $B$ in Figure~\ref{fig:exampleCm1}. While the reduction of $B$ involved both the merge and deletion rules, its index reconstruction cost, $C_m(B)$, is equal to  $1$. In fact, each node containing the variable $x_i$ with $0<i<3$ has at least a parent containing the variable $x_{i-1}$ and  a child  containing the variable $x_{i+1}$. Moreover, the node corresponding to $x_0$ has a child containing the variable $x_{1}$, and the node containing $x_{3}$ has a parent containing the variable $x_{2}$. Nevertheless, while the merge rule never increases the index reconstruction cost $C_m(B)$, the deletion rule can increase it, as shown in the reduced OBDD of Figure~\ref{fig:exampleBDD}, where the node $c$ containing the vertex $x_1$ has range  $I_c = [1, 3]$, thus its reconstruction costs 3. This means that the elimination of one of its children containing the variable $x_2$ increased the index reconstruction cost.

\section{Index-Resilient OBDDs}
\label{IR}
The analysis of the previous section has shown how the reconstruction of a corrupted index could  be quite onerous, as a consequence of the process of reduction of an  OBDD. In particular, while the merge rule never increases the overall index reconstruction cost, the application of the deletion rule could increase it. In this section, we  describe a new OBDD model where we maintain some redundancy, that is we keep some redundant nodes in the diagram, in order to guarantee  a constant index reconstruction cost  for each node. 
In particular we will define an OBDD, called {\em index-resilient reduced OBDD}, satisfying  the following properties:
\begin{enumerate}
\item the index reconstruction cost of each node $N$ is $C(N)=1$;
\item any node with a faulty index has a  reconstruction cost $O(r)$, where $r$ is the number of nodes with a corrupted index in the OBDD;
\item the pointers to the  parents of a node are never used;
\item the indices can be restored without using the unique table;
\item the new OBDD is canonical.
\end{enumerate}
Observe that, Property~3 guarantees that for the reconstruction of the index of a node $N$ we do not need to know the indices of its parents. This is very important since the number of parents of a node $N$ in a OBDD can be exponential in the number of variables; indeed, it can be $O(m)$, where $m$ is the total number of nodes in the OBDD, and, in the worst case, $m \in \Theta (2^n/n)$~\cite{LL92}.
 
Let us start with a simple observation: since the deletion rule can increase the index reconstruction cost, we could decide not to apply this rule during the reduction of an OBDD.  In this way, we have  clearly a cost $C(N) =1$ for each node $N$ in the OBDD. 
An OBDD that is reduced using exclusively the merge rule is called {\em quasi-reduced OBDD}~\cite{LL92}.

An important property of quasi-reduced OBDD is that each node at level $i$ has all parents at level $i-1$ and all children on level $i+1$.
For example, consider the OBDD in Figure~\ref{fig:quasi-red}. This OBDD has been reduced using the merge rule only.

Once we fix a variable ordering, it is easy to verify that quasi-reduced OBDDs are canonical forms.
Quasi-reduced OBDDs are an interesting solution since the growth of the number of nodes, with respect to a reduced OBDD, is not very significant, as statistically studied in~\cite{LL92}.  
Thus, quasi-reduced OBDDs are still a compact representation  and  could represent a convenient and canonical trade-off between memory saving, reduction time and error reconstruction time.

However, as we have already observed for the reduced OBDD $B$ in Figure~\ref{fig:exampleCm1},  the use of the deletion rule does not always increase the index reconstruction cost. In other words, it is still possible to delete some redundant node in a quasi-reduced OBDD guaranteeing that, in the final OBDD, the index reconstruction cost of each node  $N$ is still $C(N)= 1$.  
Most importantly, as we will show in this section, it is possible to have a {\em canonical} OBDD, {\em more compact} than a quasi-reduced one, and with  a cost $C(N) =1$ for each node $N$. 

For this purpose, we define a new class of OBDDs:
\begin{definition}[Index-Resilient OBDD] \label{IROBDD}
An {\em Index-Resilient} OBDD is an OBDD with no mergeable nodes, where each internal node $N$ on level $i$  has at least  one child on level $i+1$, for any level of the OBDD.
\end{definition}

In particular, a quasi-reduced OBDD is an index-resilient OBDD where each node on level $i$ has {\em all}  parents on level $i-1$ and {\em all} children on level $i+1$.

Observe that the {\em index reconstruction cost}  for any node $N$ in an index-resilient OBDD is  $C(N) = 1$, since the variable index of a node $N$ is directly given by $i = \min \{i_0,i_1\} - 1$ where $i_0$ and $i_1$ are the levels of the 0- and 1-child of $N$. 
Note also that for the reconstruction of the index of $N$ we do not need to know the indices of its parents (whose number is not a priori known), but only the indices of its children that are always 2 in number. 

To compute a compact index-resilient OBDD, we start from a quasi-reduced one deleting some redundant nodes while preserving the index-resilient property.
In order to efficiently test whether we can delete a redundant node $N$, 
we need the following parameter:
\begin{definition} \label{numP}
Let $B$ be an index-resilient OBDD and let $N$ be a redundant node in $B$. The parameter  $numP(N)$ is the number of parents $P$ of $N$ satisfying {\em at least one} of the following properties: 
\begin{enumerate}
\item $P.0$-child and $P.1$-child are redundant (possibly, $P.0$-child = $P.1$-child) and $N = P.1$-child;
\item $P$ has another child $N'\neq N$  on a level strictly greater than $i+1$, where  $i$ is the level of $P$. 
\end{enumerate}
\end{definition}
Note that this parameter is not defined for non redundant nodes. Moreover, if $N$ is the root and is redundant, then $numP(N)=0$.
Finally, observe that in a quasi-reduced OBDD there are no nodes $P$ satisfying the second property, as all children of any node are on the level immediately below it. 

The parameter $numP(N)$ counts the number of parents, of a redundant node $N$, whose cost is affected by the deletion of $N$. In fact, the cost $C(P)= 1$, of a node $P$ at level $i$, is not increased by the deletion of one of its children $N$ in the unique case when $P$ has the other child $N'$, on level $i+1$, that cannot be removed. The child $N'$ is not removed in two possible cases: 1) $N'$ is not redundant; 2) $N'$ is redundant (like $N$) but is the 1-child of $P$. The second criterion is an arbitrary choice due to the necessity of deleting one of the two redundant children of a node $P$ while maintaining the index reconstruction cost and the canonicity of the representation. More precisely, when a node $P$ has two redundant children, one of them can be removed without changing the cost of $P$. In this paper we always remove the 0-child of $P$ in order to guarantee that the resulting OBDD is canonical (see Theorem~\ref{canon}). 
Observe that a  redundant 0-child is not always  removed, since this node could be  a non removable one due to other parents' constraints. 
The choice of removing the 1-children is analogous.
For example, see the quasi-reduced OBDD in Figure~\ref{fig:quasi-red}. Each redundant node $N$ in the figure has a value that corresponds to $numP(N)$.

This parameter can be efficiently computed with a depth first or a breadth first visit of a quasi-reduced OBDD.  

When the quasi-reduced OBDD is constructed and $numP$ is computed, we can characterize chains of redundant nodes that can be removed, maintaining equal to 1 the index reconstruction cost of each remaining node. Consider, for example, the portion of an OBDD in Figure~\ref{fig:chain1}, the chain of redundant nodes from node $x_5$ to node $x_7$ can be removed, since the cost of the remaining nodes is not affected by the deletion. In fact, each remaining internal node on level $i$ still has, at least, a child on level $i+1$.
The same happens for the chain of redundant nodes from node $x_5$ to node $x_7$ in the  OBDD portion in Figure~\ref{fig:chain2}. Note that, in this OBDD, the node $x_5$ on the right has two distinct redundant children. The only child that can be removed is the 0-child that is in the chain. On the contrary, the chain from $x_5$ to $x_7$ in Figure~\ref{fig:chain3} cannot be completely removed because $x_7$ is a 1-child, with a redundant sibling, of the node $x_6$ on the right. In this case only the chain from node $x_5$ to node $x_6$, together with the redundant node $x_7$ on the right, can be removed. 

Our purpose is to delete chains of redundant nodes in a quasi-reduced OBDD without increasing the index reconstruction cost. We therefore introduce the concept of {\em removable chain}.
\begin{definition} [Removable chain]
\label{remC}
A {\em removable chain} in an index-resilient OBDD is a chain $C= N_1, N_2, \ldots, N_k$ (with $k\geq 1$) of redundant nodes such that:
\begin{enumerate}
\item $numP(N_1) = 0$,
\item $\forall i \in [2,\ldots,k]$\,, $ numP(N_i) = 1$,
\end{enumerate}
The node $N_1$ is called {\em head} of the chain, and the unique child $M$ of $N_k$ is called  {\em child} of the chain.
\end{definition}
The first requirement states that the head of the chain $N_1$  can only have  non redundant siblings $N'$, or redundant siblings $N' \neq N_1$ that are the 1-child of their parents. Moreover, all siblings of $N_1$ lye on the level immediately below the level of their parents. Note that this requirement implies that all parents of $N_1$ are not redundant.

The second requirement states  that the same property holds for any other node $N$ of the chain, with the only difference that now $N$ can have one redundant parent: the node above it in the chain.
As a consequence of these two conditions, only the 0-child of a node $N$ with two different redundant children  will be possibly deleted from the OBDD, while the 1-child will be kept  to maintain the node range of $N$. 
Note that when the chain is composed by a single redundant node $N$, we have that $N$ is removable when $numP(N) =0$.

\begin{figure}[t]
\begin{center}
\subfigure[\label{fig:chain1}]{\includegraphics[width=0.25\textwidth ]{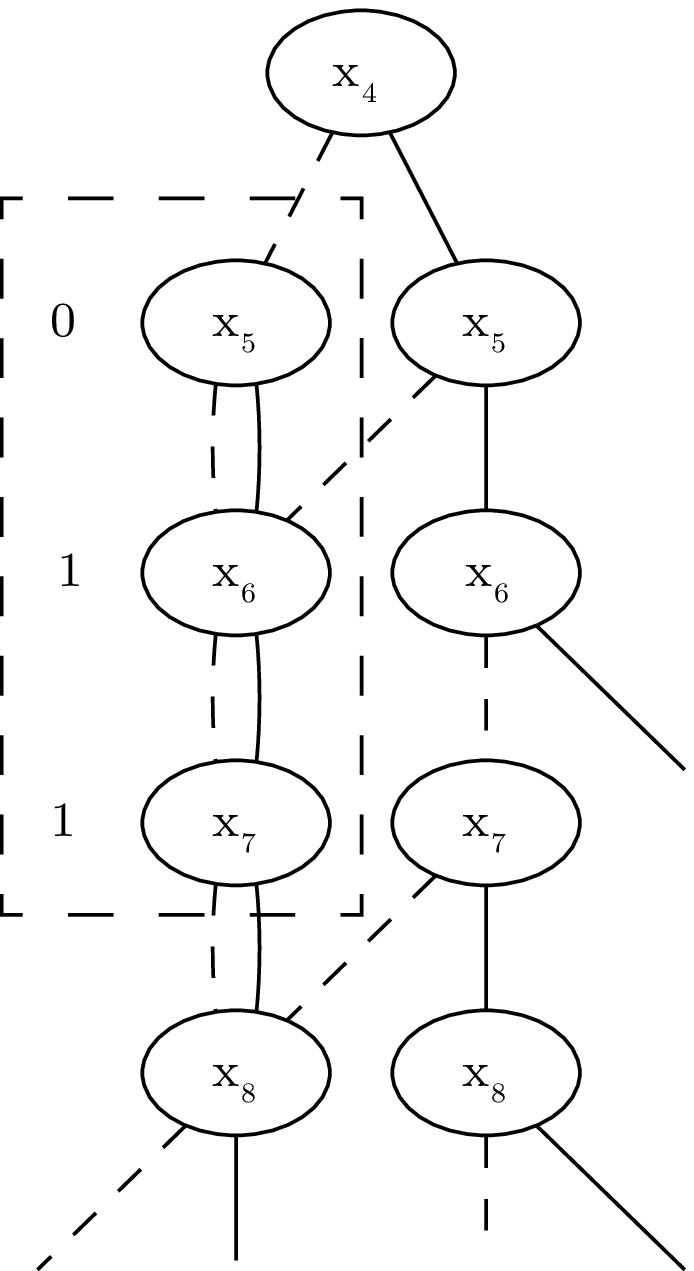}}
\hspace{6mm}
\subfigure[\label{fig:chain2}]{\includegraphics[width=0.25\textwidth ]{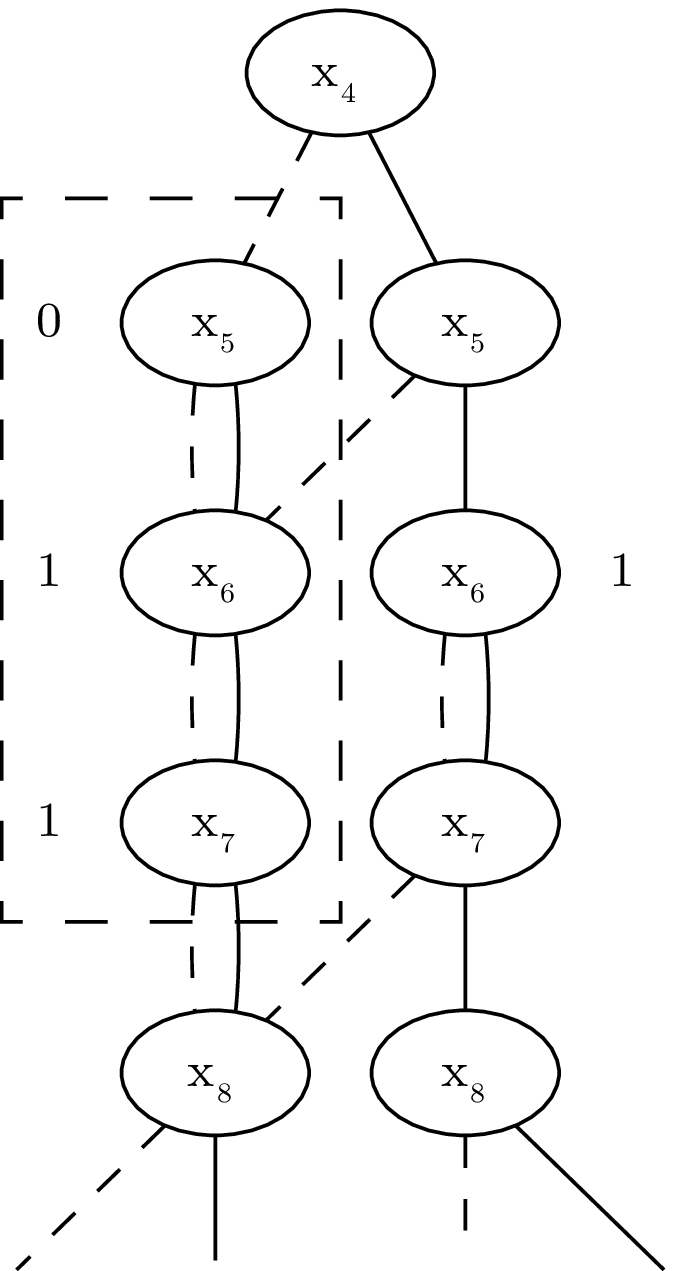}}
\hspace{6mm}
\subfigure[\label{fig:chain3}]{\includegraphics[width=0.25\textwidth ]{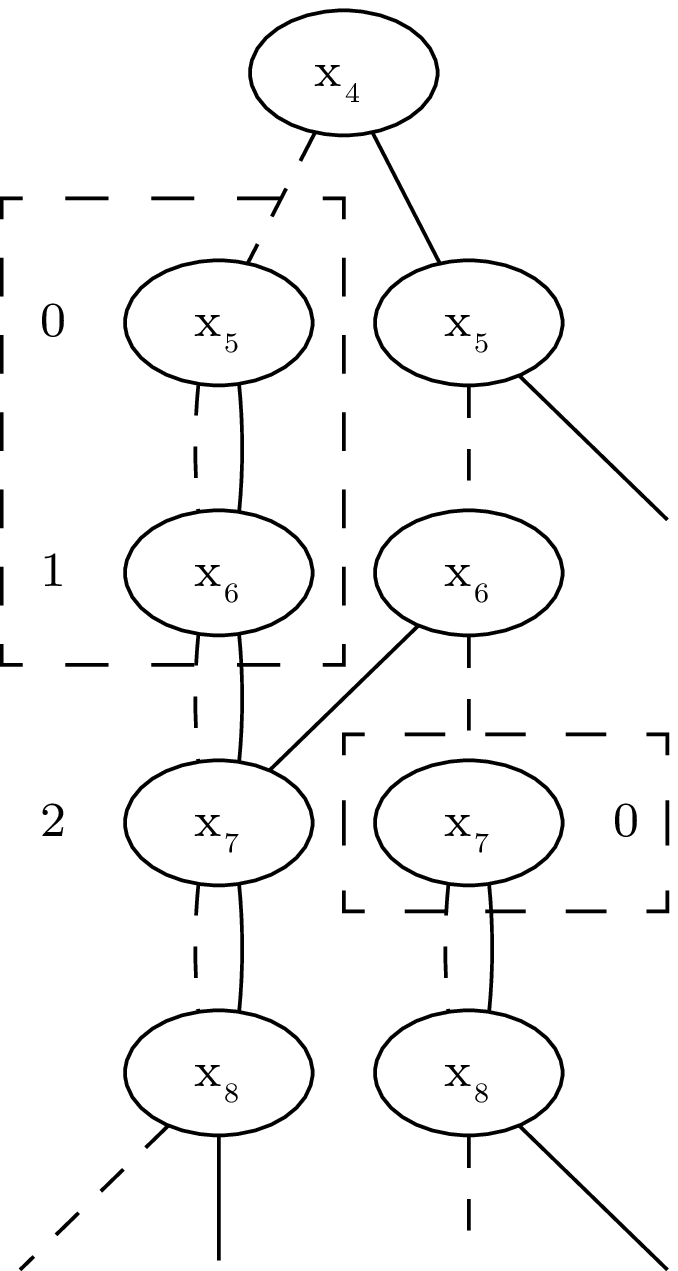}}
\end{center}
\caption{Examples of possible chains. Removable chains are in dashed boxes.}
\label{fig:chains}
\end{figure}

\begin{definition} [Maximal removable chain]
\label{MremC}
A removable chain is {\em maximal} if it cannot be further extended, i.e., the child $M$ of the chain, if redundant, does not satisfy condition 2 of Definition~\ref{remC}.  
\end{definition}

In the following proposition we show that in an index-resilient OBDD there are no maximal ``crossing'' chains. 

\begin{proposition} 
A node $N$ in an index-resilient OBDD cannot be part of two different maximal chains.
\end{proposition} 
\begin{proof}
We can observe that in an index-resilient OBDD there are no nodes with two different redundant parents. This is due to the fact that a node cannot have two redundant parents either on the same level (they would be mergeable) or on different levels, as any node must have at least a child on the level immediately below it. 
Moreover, since any internal node in an OBDD has two children (possibly the same node), there are no redundant nodes with two distinct children. 
Finally, since each chain has a head $N_1$ such that $numP(N_1) =0$ (i.e., $N_1$ has no  redundant parents), by definition of  maximal removable chain it is not possible that a maximal chain contains another maximal chain. 
From these  properties we have that any node in an index-resilient OBDD cannot be part of two distinct maximal chains. \hfill\end{proof}

The following proposition shows that the deletion of a removable chain in an index-resilient OBDD   does not change the overall index reconstruction cost, i.e., after the removal of the chain, each internal node on level $i$ still has at least a child on level $i+1$, for any level in the OBDD. 

\begin{proposition}\label{delRed}
Let $C = N_1, N_2, \ldots, N_k$, $k \ge 1$, be a removable chain in an index-resilient OBDD. The OBDD resulting from the deletion of $C$  is still index-resilient.
\end{proposition} 
\begin{proof}
We must show that the deletion of  $C$ does not change the node range of the parents of the nodes in $C$, as these are the only nodes in  the OBDD that could be affected by the deletion of $C$. More precisely, the deletion of $C$ could change the upper bound in the range of the parents. Observe that the cost $C(M)$ of the child $M$ of the chain $C$ is not increased since $C(M)$ depends on the children of $M$.

Condition 1 in Definition~\ref{remC} guarantees that the head of the chain $N_1$ has siblings that can be used to maintain the upper bound in the node range of  all its parents. Indeed, for any parent $P$ of $N_1$, the sibling $N'$ of $N_1$ lies on the level immediately below $P$ and will never be deleted from the OBDD as it is either non redundant, or redundant but not removable (since in this last case, its redundant sibling $N_1$ is  not the  1-child of $P$). 
Analogously,  condition 2 in Definition~\ref{remC} implies that each node $N_i$, $2 \le i \le k$, can only have  non redundant siblings $N'$, or redundant siblings $N' \neq N_i$ that are  the 1-child of their parents. These siblings, that  will never be deleted from  the OBDD, guarantee that the index reconstruction cost of each parent of $N_i$ remains equal to 1.
\hfill\end{proof}

The new reduction algorithm (Algorithm~\ref{fig:algoriduzioneNEW}) is based on three visits of the 0-index-resilient OBDD $B$ in input. The first visit is used to compute the parameter $numP(N)$, for each redundant node $N$ in $B$. 
Then, with a breadth first visit, all removable maximal chains are identified and their nodes are finally removed with a last visit of the OBDD, executed by the procedure Remove(). The procedure Remove() is a simple recursive depth first visit that deletes from the OBDD all nodes identified as removable.

\begin{figure}[t]
\begin{center}
\subfigure[\label{fig:quasi-red}]{\includegraphics[width=0.22\textwidth ]{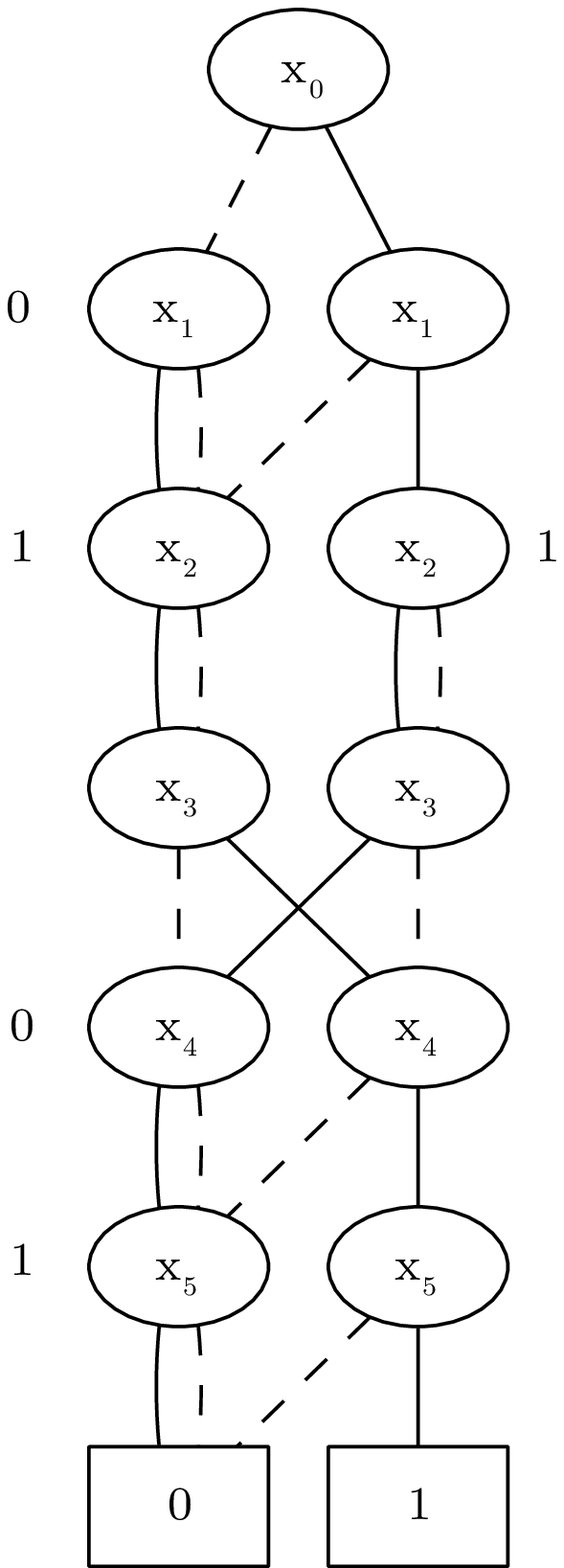}}
\hspace{20mm}
\subfigure[\label{fig:index-res}]{\includegraphics[width=0.15\textwidth ]{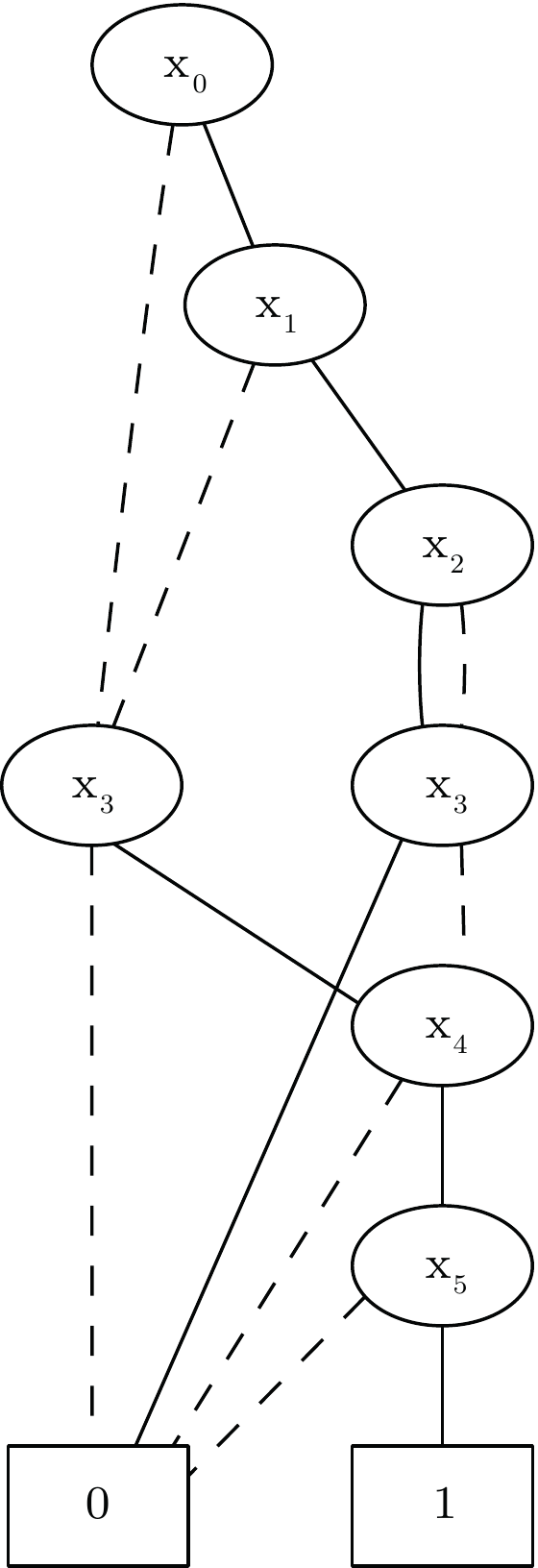}}
\end{center}
\caption{Example of the reduction algorithm. On the left a quasi-reduced OBDD and, on the right, the corresponding index-resilient reduced OBDD.}
\label{fig:exampleAlgo}
\end{figure}

\begin{algorithm}[Reduction algorithm for  index-resilient OBDD]
\label{fig:algoriduzioneNEW}

\ \ 

\begin{scriptsize}
\vspace{0.15cm} 
\hrule 

\vspace{-0.3cm}

\begin{tabbing}
{\bf INPUT}\\
\var{B} \hspace{1.5cm}\= /* index-resilient OBDD to be reduced */\\
{\bf OUTPUT}\\
\var{IRR-B}\uno /* index-resilient reduced OBDD */\\ \\
{\bf MAIN} \\
\, \com{for} \= \com{each} node $N \in B$ \com{do} \\
\uno $ToRemove(N) = False$; \\
\uno \com{if} \=  $(N.0$-$child == N.1$-$child)$ /* $N$ is a redundant node */\\
\due $numP(N)=0$; \\
\, \com{for} \= \com{each} node $N \in B $  \com{do}  /* computation of $numP$ (Definition~\ref{numP})*/\\
\uno \com{if} \=  $(((N.0$-$child).0$-$child == (N.0$-$child).1$-$child) \&\& ((N.1$-$child).0$-$child == (N.1$-$child).1$-$child))$\\
\due       $numP(N.1$-$child)$++;  \ \ \ \ /*$N.1$-$child$ and $N.0$-$child$ are redundant (Definition~\ref{numP}.1) */ \\
\uno \com{if} \= $((N.0$-$child).index$ $ != N.index +1)$ \\
\due       $numP(N.1$-$child)$++;  \ \ \ \  /*$N.0$-$child$ is on a level $> N.index +1$ (Definition~\ref{numP}.2) */ \\ 
\uno \com{if} \= $((N.1$-$child).index$ $ != N.index +1)$ \\
\due       $numP(N.0$-$child)$++;  \ \ \ \  /*$N.1$-$child$ is on a level $> N.index +1$ (Definition~\ref{numP}.2) */ \\ 
\, $nl=$ nLevels(B);  /* $nl$ is the number of levels in the $B$ */\\
\, \com{for} \= $(i=0; i< nl-1; i$++$)$ /* breadth first visit for the deletion of removable nodes */ \\
\uno \com{for} \= \com{each} node $N$ at level $L_i$  \com{do}  \\
\due \com{if} \= $((N.0$-$child == N.1$-$child)\&\&(numP(N) == 0)$ /* if $N$ is a head of a  removable chain */ \\
\3 RemovableChain($N$);  /* find any node $N$ in a rem. chain and set $ToRemove(N)$ to $True$*/\\
\, IRR-B= Remove($B$); /* remove from $B$ any node $N$ such that $ToRemove(N)$ is $True$*/\\ 
\, \com{return} IRR-B;
\end{tabbing}
\begin{tabbing}
{\bf RemovableChain($N$)} \\
\, $end =$ \com{False}; \\
\, \com{wh}\=\com{ile}   $(!end)$  /* the chain is not finished */ \\
\uno $ToRemove(N) = True$;\\
\uno $N = N.0$-$child$;\\
\uno \com{if} \= $(isLeaf(N) || (N.0$-$child \neq N.1$-$child) || (numP(N) > 1))$ /* the removable chain is finished*/\\
\due $end =$ \com{True}; \\
\end{tabbing}
\vspace{-20pt} \hrule
\end{scriptsize}
\end{algorithm}

Recall that, when we compute the parameter $numP$ starting with a quasi-reduced OBDD, we only have to consider the first property in Definition~\ref{numP} (i.e., the first \com{if} in the second \com{for each} of  Algorithm~\ref{fig:algoriduzioneNEW}). 

The correctness of the new reduction algorithm is proved in the  following theorem.  
\begin{theorem}\label{rid}
Let $B$ be an index-resilient OBDD. Algorithm~\ref{fig:algoriduzioneNEW}, with input $B$, computes a index-resilient OBDD $B_r$ equivalent to $B$ that does not contain any removable chain.
\end{theorem} 
\begin{proof}
First, observe that the new reduction algorithm modifies the input OBDD $B$ only applying the deletion rule to a subset of its redundant nodes. Thus, the resulting OBDD $B_r$  is equivalent to $B$.

Second, we can observe that the algorithm removes only maximal chains. In fact, the algorithm finishes the construction of a removable chain  with a redundant child $M$ (that is not removed) only if $numP(M) > 1$. This means that there is another parent of $M$, not in the chain, that has two different redundant children and $M$ is its 1-child. Thus, the deletion of the chain cannot make $M$  removable.

To complete the proof, we must show that the deletion of a  maximal chain $C=N_1, N_2, \ldots, N_k$ cannot make redundant, and therefore possibly removable,  the parents (not in $C$) of any  node $N_i$, for any $1 \le i \le k$. 
For this purpose, we can observe that when we remove a chain of redundant nodes, the parents of these nodes, that are not included in the chain, become parents of the child $M$ of the chain.
By definition of removable chain, each parent $P$ of $N_i$, outside the chain, has another child $N'\neq N_i$ on level $\ell+1$, where $\ell$ is the level of $P$. Once the chain has been deleted, each parent  $P$ of any node  $N_i$ ends up with two children on different levels: $N'$  on level $\ell + 1$ and $M$ on a level strictly greater than $\ell +1$, since $M$ is a descendant of a child of $P$;  thus none of these parents can become redundant. 
\hfill\end{proof}

The cost of the algorithm is linear in the size of the OBDD in input, as it basically consists in just three visits of the data structure and 
each chain is visited only once, starting from its head. The $Remove(BDD)$ procedure removes a node $N$ only if $ToRemove(N)$ is $True$. The cost of the reduce procedure is linear in the number $m$ of nodes  in the OBDD $B$ (note that, since any internal node has two children, the number of edges in a OBDD is $O(m)$).

\begin{example}
Consider the quasi-reduced OBDD in Figure~\ref{fig:quasi-red}.  Algorithm~\ref{fig:algoriduzioneNEW}, starting from a quasi-reduced OBDD, first computes $numP(N)$ for each redundant node $N$. The second visit of the OBDD is a breadth first visit that considers the redundant nodes in each level and checks whether they are  heads of removable chains, starting from the root (level $L_0$). Starting from the head of any removable chain, the algorithm decides if each node $N$ in the chain can be removed, checking its parameter $numP(N)$.
In the example, the first chain considered is the one that starts with the 0-child $N_1$ of the root, which has $numP(N_1) = 0$ (condition 1 in Definition~\ref{remC}). Its unique child $N_2$ (corresponding to the variable $x_2$) is also redundant, and it is such that $numP(N_2) =  1$ (condition 2 in Definition~\ref{remC}),  moreover its child $N_3$ (corresponding to the variable $x_3$) is not redundant. Note that $x_2$ has a redundant sibling but it is the 0-child of their common parent.  The node $N_3$ is the child of the maximal chain. 
Then, the algorithm considers the next redundant node  not yet visited, that is the node with label $x_2$ on the right of   level $L_2$. This node $S$ cannot be head of a chain since $numP(S) \ne 0$; indeed, $S$ is the 1-child of a node with two redundant children.
The algorithm, finally, takes into account the redundant node $Q$, corresponding to the variable $x_4$ on level $L_4$. $Q$ can be the head of a removable chain, as $numP(Q) = 0$. Its unique child $R$ is also redundant and such that  $numP(R) = 1$. 
Thus, $Q$ and $R$ form a maximal removable chain.
The resulting OBDD, shown in Figure~\ref{fig:index-res}, is 0-index-resilient and does not contain any removable chain.
\end{example}

We now formally introduce the concept of {\em Index-Resilient Reduced OBDD}.

\begin{definition}[Index-Resilient Reduced OBDD]
An index-resilient \\ OBDD is {\em reduced} if it does not contain any removable chain.
\end{definition}

In the next theorem we summarize and prove some important properties of the index-resilient reduced OBDDs obtained with the proposed reduction algorithm applied on a quasi-reduced OBDD. 
 
\begin{theorem} 
\label{canon}
Let $B$ be an quasi-reduced OBDD and let $B_r$ be the index-resilient reduced OBDD obtained with the new reduction algorithm with input $B$ (Algorithm~\ref{fig:algoriduzioneNEW}).
Then
\begin{enumerate}
\item for each node $N$ in $B_r$, $C(N) = 1$;
\item $B_r$ does not contain mergeable nodes;
\item $B_r$ is canonical, i.e.,  given a function $f$ and a variable ordering $<$, $B_r$ is the only index-resilient reduced OBDD with variable ordering $<$ that represents $f$.
\end{enumerate}
\end{theorem}
\begin{proof}
\begin{enumerate}
\item  Since the algorithm starts with a quasi-reduced OBDD, Proposition~\ref{delRed} guarantees that the deletion of  the removable chains maintains the OBDD index-resilient, i.e., for each remaining internal node on level $i$ there exists at least a child on level $i+1$, for any level of the OBDD.
\item Suppose by contradiction that  $N$ and $N'$ are two mergeable nodes at level $i$ of $B_r$. $N$ and $N'$  have a child at level $i+1$, and a child $N_j$ at level $j >i$. If $j = i+1$, the fact that $N$ and $N'$ are mergeable is in contradiction with the fact that the algorithm started with a quasi-reduced OBDD. If $j >i+1$, then  the original quasi-reduced OBDD contained a removable chain between $N$ and $N_j$ and one between $N'$ and $N_j$. This in turns implies that $N_j$ had two redundant parents on the same level in the original quasi-reduced OBDD. Thus we reach a contradiction with the fact that the starting OBDD is quasi-reduced, as the  two redundant parents of $N_j$ are mergeable.
\item Given a quasi-reduced OBDD there is an unique way to delete maximal reducible chains, since each node cannot be part of two different chains and the removal of reducible chains does not produce new removable chains or mergeable nodes.  
Here it is important to recall that there is no ambiguity in the definition of removable chains: given any node $P$ with two redundant and different children, only the 0-child  can be head  or part  of a removable chain. 
The  index-resilient reduced OBDD is then a canonical form.
\end{enumerate}\hfill\end{proof}

In summary, index-resilient reduced OBDDs represent a good trade-off between index reconstruction cost and number of nodes in the OBDD. 
Moreover, the index reconstruction cost remains limited even in presence of more than one error on the indices, as stated and proved in the following theorem.
\begin{theorem}\label{Th:recostr} 
The reconstruction cost of a node $N$ on level $i$ in a index-resilient reduced OBDD $B$ affected by $r$ errors on the indices is $O(\min (r, 2^{n-i}))$. 
\end{theorem}
\begin{proof}
First, we note that, starting from the node $N$ on level $i$, there is always a complete path (i.e., a path containing the variable $x_i, x_{i+1},\ldots,x_n$) that ends on a leaf. This path can be exploited to reconstruct the index of $N$.
In fact, the index of  $N$ can be computed using the indices of its children in the following way. Let $j_0$ and $j_1$ be the levels of the 0-child and of the 1-child of $N$. If both children are not affected by errors, then the index of $N$ is $i = \min \{j_0, j_1\} - 1$. 
Otherwise, we recursively proceed on the OBDD rooted in any corrupted child of $N$, and we will restore the index of $N$ when both the indices of its children will be corrected. The recursion stops on corrupted nodes with two  uncorrupted children.
Note that we can consider the two terminal nodes (the leaves of the OBDD) uncorrupted, as they could be memorized in a safe memory, or duplicated. In the worst case, the reconstruction cost is the minimum between the dimension of the OBDD rooted in $N$ (i.e., $O(2^{n-i})$) and the total number of corrupted nodes in $B$ (i.e., $O(r)$).
The number of visited nodes is then $O(\min (r, 2^{n-i}))$.
\hfill\end{proof}

We can observe that this strategy does not exploit the unique table of the OBDD.

Finally observe that, even if in our analysis we have implicitly assumed  that an OBDD is constructed correctly, and that  memory faults occur when the data structure is in use, this assumption can be completely removed for index-resilient reduced OBDDs. Indeed, their construction starts from a binary decision tree that is transformed into a QR-BDD applying the merge rule, and in both models each node has all children on the level immediately below. Moreover, during the execution of the reduction algorithm on a quasi-reduced OBDD, we always guarantee that each node has at least one child on the level below, thus a faulty index can be immediately detected and restored. 

\section{Operations on Index-Resilient OBDDs}
\label{oper}

In the previous section we have described a new OBDD data structure that is resilient to index faults and we have shown how  an index-resilient (reduced) OBDD can be constructed starting from a quasi-reduced one.

However, OBDDs are not often constructed from quasi-reduced ones, but instead through a sequence of binary Boolean operators (as AND, OR, EXOR) applied to other OBDDs mainly using the algorithm  {\sc Apply}  reviewed in the Appendix. Therefore,  we now discuss how this   algorithm can be 
modified in order to guarantee that the OBDD in output is resilient to index faults. Moreover, the new described algorithm  for {\sc Apply} is error resilient itself, i.e., it computes an OBDD resilient to index faults even if some errors occur during the computation.  In other words, in this section we consider a dynamic framework where OBDDs are dynamically computed in a non-safe memory.

When performing operations on OBDDs, we will always suppose that the OBDDs in input are index-resilient reduced OBDDs, and we will show that the OBDD in output is still an index-resilient reduced OBDD.

\subsection{Error Correction Procedures}
\label{errorcorrection}
The algorithm discussed in this section derives from the standard algorithm {\sc Apply} avoiding the use the unique table. Moreover, as described in the Appendix, the polynomial complexity of the {\sc Apply} algorithm is due to a matrix $M_A$ that contains the pointers to the nodes computed in the recursive calls.
Therefore, we first discuss how to handle errors occurring in the data structures used by OBDD operators, i.e., the indices in the OBDD and the matrix $M_A$.

\paragraph{Errors in Indices}
We first recall that, in an  index-resilient OBDD, each internal node $N$ on level $i$  has at least one child on level $i+1$. The  algorithm for the reconstruction of a faulty index is described in the proof of Theorem~\ref{Th:recostr} and is shown in Algorithm~\ref{fig:indrec}. In particular, the reconstruction is based on the following property.  Let  $N$ be a node on level $i$ of an index-resilient OBDD, there is always a complete path (i.e., a path containing the variable $x_i, x_{i+1},\ldots,x_n$) from $N$ to a leaf.
The complete strategy is described in the recursive function {\bf IndexReconstruct(B,N)} (see Algorithm~\ref{fig:indrec}). 
Recall that, by Theorem~\ref{Th:recostr}, the number of recursive calls of Algorithm~\ref{fig:indrec} is in $O(r)$, where $r$ is the total number of errors in the index-resilient OBDD.

\begin{algorithm}[Index reconstruction]
\label{fig:indrec}

\ \ 

\begin{scriptsize}
\vspace{0.15cm} 
\hrule 

\vspace{-0.3cm}

\begin{tabbing}
{\bf INPUT}\\
\var{B} \hspace{1.5cm}\= /* index-resilient OBDD with $n$ variables */\\
\var{N} \hspace{1.5cm}\= /* node in $B$ with a faulty index */\\
{\bf OUTPUT}\\
\var{index}\uno /* the correct index of $N$ */ \\
{\bf SIDE EFFECTS} \\
$N$ has been corrected in  $B$ (and, possibly,   some paths from $N$ to leaves have been corrected)\\ \\
{\bf IndexReconstruct(B,N)} \\
\, \com{if} \=  (($N.0$-$child$ is a leaf) $\&\&$  ($N.1$-$child$ is a leaf)) /* $N$ has index $n-1$ */\\
\uno $N.index = n-1$; \\
\, \com{else if}    ($N.0$-$child$ is a leaf) \\
\uno \com{if} \= (($N.1$-$child).index$ is not correct)\\
\due ($N.1$-$child).index$  = IndexReconstruct(B,$N.1$-$child$); \\
\uno $N.index = (N.1$-$child).index -1$; \\
\, \com{else if}    ($N.1$-$child$ is a leaf) \\
\uno \com{if} \= (($N.0$-$child).index$ is not correct)\\
\due ($N.0$-$child).index$  = IndexReconstruct(B,$N.0$-$child$); \\
\uno $N.index = (N.0$-$child).index -1$; \\
\, \com{else} /* both children are not leaves */\\
\uno \com{if}  (($N.0$-$child).index$ is  not correct))\\
\due ($N.0$-$child).index$  = IndexReconstruct(B,$N.0$-$child$); \\
\uno \com{if}  (($N.1$-$child).index$ is  not correct))\\
\due ($N.1$-$child).index$  = IndexReconstruct(B,$N.1$-$child$); \\
\uno $N.index = \min\{(N.0$-$child).index, (N.1$-$child).index\}-1$; \\
\, \com{return} $N.index$;
\end{tabbing}
\vspace{-8pt} 
\hrule
\end{scriptsize}
\end{algorithm}
\paragraph{Errors in Matrix $M_A$}
Matrix $M_A$ is exploited by the {\sc Apply} algorithm  for memorizing the output of the recursive calls in order to avoid an exponential number of re-computations (see the Appendix for more details). In particular, let  $B$ and $B'$ be the two OBDDs whose roots are the inputs to the {\sc Apply} algorithm, we have that matrix $M_A$ is a table with $t$ rows and $t'$ columns, where $t$ (resp.,  $t'$) is the number of nodes in the OBDD $B$ (resp.,  $B'$). Position $M_A[N,M]$ contains NULL if {\sc Apply}$(N,M)$ is never computed, otherwise,  $M_A[N,M]$ contains the pointer to the root node of the sub-OBDD that is the solution of {\sc Apply}$(N,M)$. 
If $M_A[N,M]$ is corrupted the algorithm simply computes {\sc Apply}$(N,M)$. This means that, in the worst case, the number of re-computations of the same pointers in $M_A$ is $O(r_{M_A})$ where $r_{M_A}$ is the number of errors in $M_A$.

\subsection{The {\sc Apply} Algorithm}
\label{apply}

Let us now discuss how to execute the {\sc Apply} procedure  on  two index-resilient reduced OBDDs, representing two functions $f$ and $g$, and compute a new  OBDD $B$, representing $(f \mbox{ {\bf op} } g)$ for a given Boolean binary operator {\bf op}. The obtained OBDD $B$ will then be transformed in an index-resilient reduced OBDD through a reduction procedure described below.

Unlike the standard implementation, this {\sc Apply} procedure makes use of the matrix $M_A$  without exploiting  the unique table. Therefore, the OBDD in output can contain mergeable nodes. Observe that the matrix $M_A$ is enough to guarantee that the  {\sc Apply} procedure is quadratic, and that the size of the output OBDD is bounded by the product of the size of the two input OBDDs.
The second difference with respect to the standard implementation is that during the execution of the algorithm we do not delete redundant nodes in order to preserve the index resilience property. As already mentioned, the detection and deletion of the removable chains is executed in a second step.

We can now observe that the {\sc Apply} algorithm, executed on two index-resilient OBDDs,  preserves, by construction, the property that each internal node on level $i$, of the output OBDD,  has at least one child on level $i+1$.
Indeed, recall that if {\sc Apply} is called on two nodes $N$ and $M$ with the same index $i$, then a new node $U$ with index $i$ is created, and the algorithm is recursively executed on the two 0-children and on the two 1-children of $N$ and $M$ to generate the OBDDs whose roots become  the $0$-child and $1$-child of $U$,  respectively.
If instead $N$ and $M$ have different indices, {\sc Apply} creates a new node $U$ with the {\em lowest} index between those of $N$ and $M$, and proceeds recursively by pairing the $0$ and $1$-child of the node with lowest index with the  other node  to generate the OBDDs whose roots become the $0$-child and $1$-child of $U$,  respectively.

Without loss of generality, suppose that $N$.index $=i$ and $M$.index $\ge i$.
Then the new node $U$ has index $i$ and at least one of its children has index $i+1$, since 
\begin{enumerate}
\item {\sc Apply} recurs on the two children of $N$, one of which  certainly has index $i+1$, paired with $M$ (if $M$.index $> i$) or  with the corresponding children of $M$ (if $M$.index $> i$), whose indices are in both cases greater or equal to $i+1$;
\item {\sc Apply} always chooses as index of the new node the lowest index between  those of the two nodes in input.
\end{enumerate}

During the execution of the {\sc Apply} procedure, some errors in the data structures may occur, in particular errors in the indices of the nodes in the OBDDs in input and errors in the recursive table  $M_A$.
All these errors can be handled as explained in the previous Section~\ref{errorcorrection}.

\begin{figure}[t]
\begin{center}
\subfigure[Algorithm~\ref{fig:algoriduzioneNEW}: input. \label{fig:piccolo}]{\includegraphics[width=0.405\textwidth ]{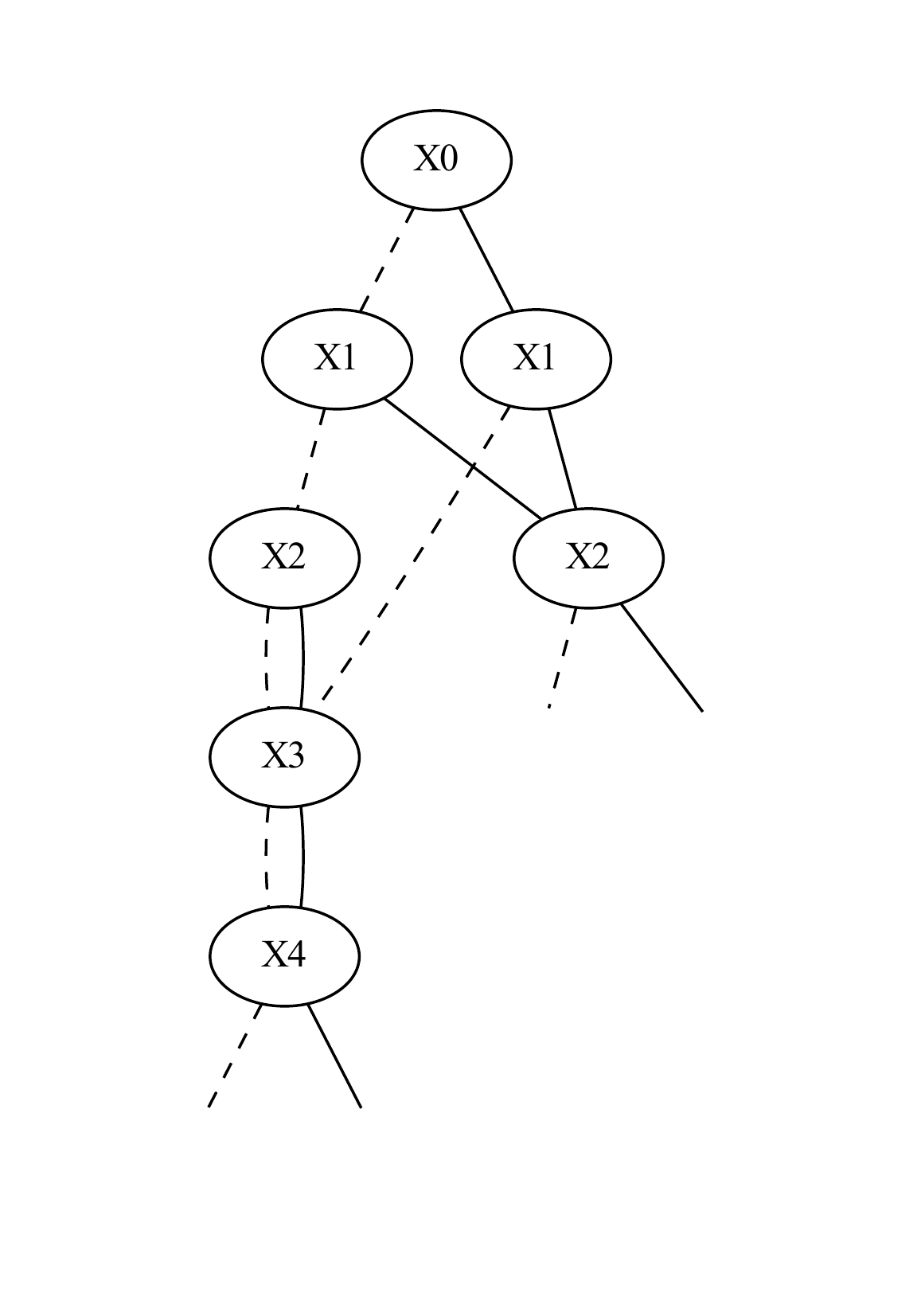}}
\subfigure[Algorithm~\ref{fig:algoriduzioneNEW}: output.\label{fig:problema}]{\includegraphics[width=0.40\textwidth ]{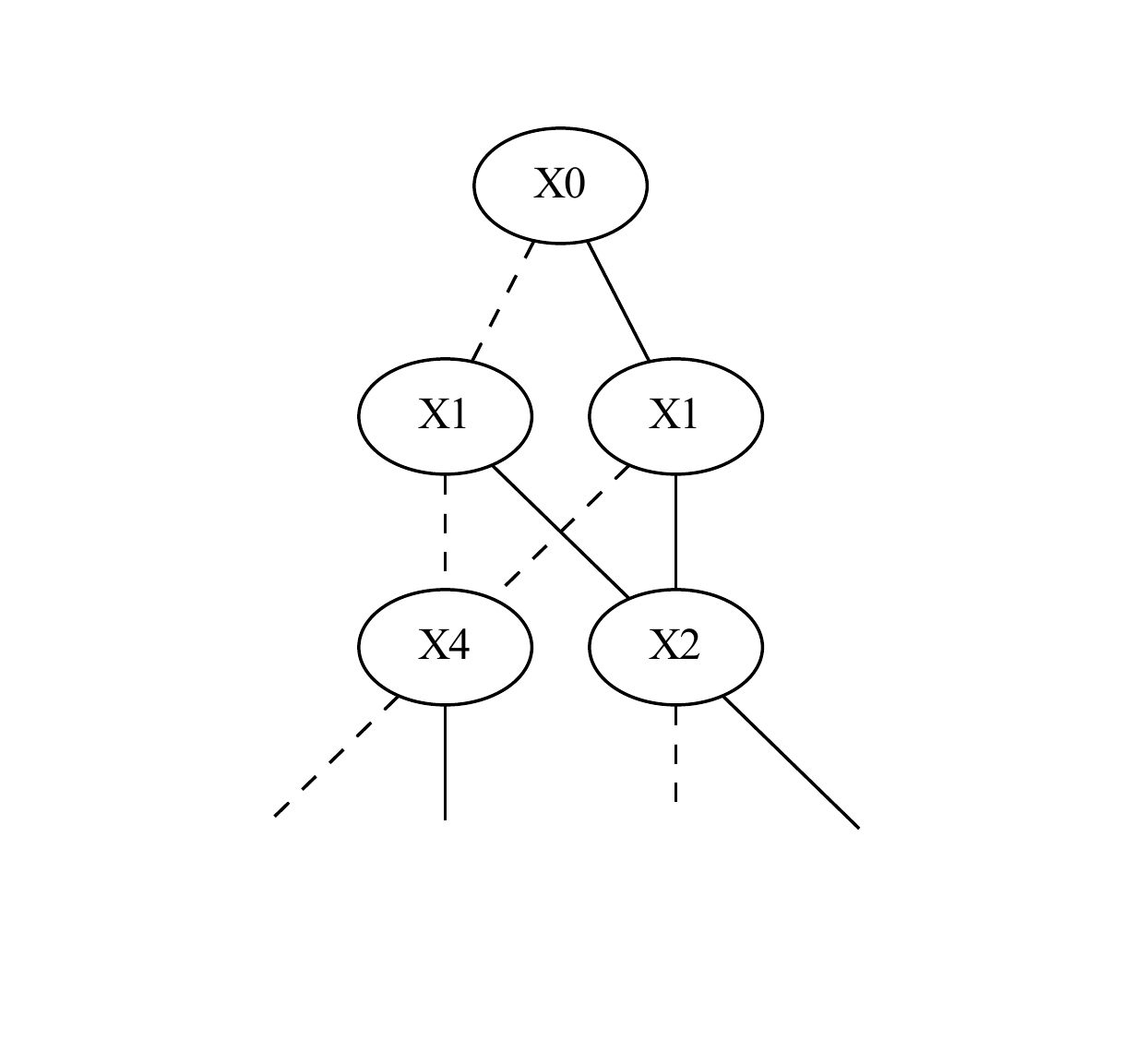}}
\caption{Example of execution of Algorithm~\ref{fig:algoriduzioneNEW} onto an index-resilient OBDD (a) that is not a quasi-reduced one: the resulting OBDD (b) contains two mergeable nodes.}
\label{piccoloProblema}
\end{center}
\end{figure}

The OBDD $B$ for $(f \mbox{ {\bf op} } g)$ computed by {\sc Apply} can contain mergeable nodes and chains of removable redundant nodes.
In order to get an index-resilient reduced  OBDD we must perform some operations on the data structure.
First of all, observe that we cannot simply reduce  $B$ by running the reduction algorithm  described in Section~\ref{IR}. Indeed, since $B$ is not a quasi-reduced OBDD, Algorithm~\ref{fig:algoriduzioneNEW} may introduce some new mergeable nodes in the OBDD, as shown in the example depicted in Figure~\ref{piccoloProblema}. 
Therefore, instead of running Algorithm~\ref{fig:algoriduzioneNEW} alone, we first transform 
$B$ into a quasi-reduced  equivalent OBDD, and we  perform the following {\sc Reduction} procedure:
\begin{enumerate}
\item {\em Transformation into a quasi-reduced OBDD.} For any $i$ and for any node $N$ on level $i$ with one  child $M$ on level $j > i+1$, we insert a chain of redundant nodes between $N$ and $M$, with consecutive indices of value in the range $[i+1, j-1]$. This operation is simply implemented with a visit of the OBDD $B$.
After this step, all paths from the root of $B$ to the terminal nodes have length $n$. Also observe that the size of the OBDD increases in the worst case for a multiplicative factor of order $n$.
\item {\em Merge.} Since the obtained OBDD $B$ contains mergeable nodes we need to reduce it applying the merge rule to its nodes before executing Algorithm~\ref{fig:algoriduzioneNEW}. In order to guarantee the error resilience of this procedure we avoid the use of unique tables or other data structures, at the  expense of a quadratic time complexity.
In fact, this task can be accomplished through a DFS visit of the OBDD:  when a node $N$ is visited, we perform a second visit that identifies and merge all nodes mergeable with $N$. 
\item {\em Removal of redundant chains.} Finally, we can reduce the OBDD $B$ applying Algorithm~\ref{fig:algoriduzioneNEW}. 
\end{enumerate}
As for the {\sc Apply} procedure, if during the execution of these steps some errors in the indices occur, we handle them as explained in  Section~\ref{errorcorrection}.

After all these operations, the final OBDD will be an index-resilient reduced OBDD, as shown in the following theorem.
\begin{theorem}\label{theoApply}
Let $B$ be the  OBDD computed executing the {\sc Apply} algorithm 
and the {\sc Reduction} procedure. Then $B$ is an index-resilient reduced OBDD.
\end{theorem} 
\begin{proof}
In order to prove that $B$ is index-resilient, we must show that $B$ does not contain mergeable nodes, and that  each internal node $N$ on level $i$  has at least  one child on level $i+1$. This follows immediately since steps 1 and 2, of the {\sc Reduction} procedure, transform $B$ into a quasi-reduced OBDD. Thus,  Algorithm~\ref{fig:algoriduzioneNEW} is executed on a quasi-reduced OBDD and Theorem~\ref{canon} implies that $B$ is index-resilient and reduced.
\hfill\end{proof} 

Finally, observe that if we use step 2 as the merge strategy to build a quasi-reduced OBDD starting from a binary decision tree (see Section~\ref{IR}), we have that steps 2 and 3 of the {\sc Reduction} procedure provide an error resilient algorithm for the construction of index-resilient reduced OBDDs starting from binary decision trees. 
 
\section{Error in Edges}
\label{edges}
Besides indices, it is important to study how to correct errors in pointers to the children of a node. Pointer based data structures are not, in general, error resilient, since the loss of a pointer can imply the loss of an important part of the data structure.
This is the scenario that we have for OBDDs. In fact, if a node $N$ has a unique parent and an error occurs to the pointer from the parent to $N$, then $N$ and, possibly, part of the subgraph rooted in $N$ are no more reachable.
Moreover, the loss of a pointer can give errors even if the pointed node has more than one parent, as discussed in the following example.

\begin{figure}[t]
\begin{center}
\includegraphics[width=0.4\textwidth ]{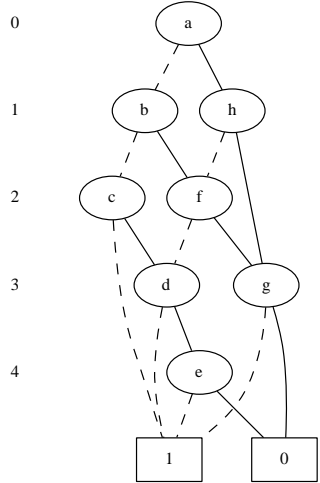}
\end{center}
\caption{\label{fig:exbddegdes}OBDD described in Example~\ref{ex:edges}.}
\end{figure}

\begin{example}
\label{ex:edges}
Consider the OBDD in Figure~\ref{fig:exbddegdes} and the node $c$. Suppose that its 1-edge, i.e., the edge pointing to $d$,  is corrupted. Note that, although we are loosing the edge from $c$ to $d$, $d$ is still connected to the rest of the OBDD through the edge from $f$ to $d$. Nevertheless, we cannot decide which is the value of the function represented by the OBDD in the case we are considering a path through $c$ and the variable contained in $c$ is true.
\end{example} 
    
Observe that in the unique table we have always a pointer to each node. The problem is to reconstruct a corrupted  link between two nodes in the OBDD.
The unique table can be useful in this reconstruction, but, unfortunately, in the unique table we do not memorize the pointer values, but only the indices. In fact, we use the hash of the pointer values to identify a node.  More precisely, consider the unique subtable $\mathcal{U}_{i}$ corresponding to the variable  $x_i$ of the unique table $\mathcal{U}$ of a given OBDD. We can find the subtable $\mathcal{U}_{i}$ of $\mathcal{U}$ knowing the index $i$. In $\mathcal{U}_{i}$ we have all the nodes that contain  $x_i$. 
Each node is memorized using an hash function on its children. Consider the node  $N$ characterized by the tuple $[i, N.\mbox{0-child}, N.\mbox{1-child}]$:  we have that $N \in \mathcal{U}_{i}$ and its index in the subtable $\mathcal{U}_{i}$ is the integer value $Hash(N.\mbox{0-child}, N.\mbox{1-child})$. Obviously, starting  from such index value we cannot directly compute the values $N.\mbox{0-child}$ and $N.\mbox{1-child}$.

A possible way to use the unique table is to try all the possible pointers of the graph. In particular, let $N$ be a node, corresponding to the tuple  $[i, N.\mbox{0-child}, N.\mbox{1-child}]$, with an error in $N$.1-child. For reconstructing $N$.1-child we can try each node $M$ (where $M$ is the pointer to the node) of the OBDD and verify all the nodes corresponding to $Hash(N.\mbox{0-child}, M)$ on the subtable $\mathcal{U}_{i}$; if one of them is the pointer to $N$ we have found that $M$ corresponds to $N$.1-child. 

This strategy is very expensive since, in the worst case, we would check each node in the graph. To restrict the number of possible checks, we can exploit again the OBDD properties. If we have a node $N$ at level $l_N$ (i.e., with index $x_{l_N}$), each of its  children must have a level $l_C > l_N$. We have then a bound on the number of possible nodes to check. Of course, if we have an error in the levels close to the OBDD root, the number of nodes to check is still high. 

A possible way to improve the pointer reconstruction could be that of storing all pointers to the nodes in an additional  vector, in the order given by a depth first visit of the OBDD, following first the 0-edges, as  suggested in~\cite{SJ13}. 
Therefore, let us assume that the nodes of the OBDD are saved contiguously in the memory, in the order given by a depth first visit of the OBDD, following first the 0-edges; alternatively we can assume that all pointers to the nodes are stored, in the same order, in an additional  vector.
Each node $N$ is represented in the vector as the triple [index, $N$.0-child, $N$.1-child].
In contrast to what was done in~\cite{SJ13}, we will consider terminal nodes as simple nodes, and we will assume that the terminal node with label $0$ has index $\underline{0}$ and pointers to NULL, while the terminal node with label $1$ has index $\underline{1}$ and pointers to NULL; e.g., the $0$ terminal node is represented as $[\underline{0} - -]$.

\begin{definition}[Node vector] 
The {\em node vector} $\mathcal{V}_B$ of an OBDD $B$ is a vector containing all nodes $N \in B$ in the order given by a depth first visit of $B$, following first the 0-edges.
\end{definition}
For instance, the node vector of the OBDD $B$ in Figure~\ref{fig:exbddegdes}, is given by  $$\mathcal{V}_B = \text{[0 b h] [1 c f] [2 \underline{1} d] [3 \underline{1} e] [\underline{1} - -] [4 \underline{1} \underline{0}] [\underline{0} - -] [2 d g] [3 \underline{1} \underline{0}] [1 f g]}.$$

In~\cite{SJ13} the main goal is to reduce the size of the OBDD, so  that all redundant information are then deleted from the vector of nodes, with the final effect that some pointers are removed from the memory. 
In fact, given a node $N$, if $N$.0-child has not been visited yet, then it will be surely stored in the memory area adjacent to that of $N$. Analogously, if $N$.1-child  has not been visited yet, then it will be stored after the subgraph rooted in  $N$.0-child. Therefore, 
the pointers $N$.0-child and $N$.1-child are not necessary anymore, as the children of $N$ can be found computing their position.

\begin{example}
Consider the first node of the OBDD $B$ in Figure~\ref{fig:exbddegdes}. As both its children have not yet been visited, we do not need  their identification to compute their positions. On the contrary, e.g., for the node $f$ we need the pointer to $d$, as $d$ has already been visited and stored before $f$ and its position cannot be computed. For the last node visited, $h$, we  need to keep both pointers.
The OBDD could then be represented in a more compact way by the vector  $$\text{012\underline{1}3\underline{1}4\underline{1}\underline{0}2d3\underline{1}\underline{0}1fg}.$$
\end{example}

Observe that in this way it is possible to quickly find the children for some nodes, but it is not possible to determine a priori which nodes will have this characteristic, i.e., for which corrupted node it will be possible to recompute the children.
Thus, the idea is to use the computed value as an upper bound, i.e., as a position in the node vector over which the child cannot be found.
\begin{definition}[Child bound]
Given a node vector  $\mathcal{V}$, a node $N$, and a child $N_f$ of $N$, the  {\em child bound}  is the integer value $L_{N_f}$ s.t. $$N_f \in \{\mathcal{V}[0], \mathcal{V}[1], \ldots, \mathcal{V}[L_{N_f}]\}.$$
\end{definition}
The value of $L_{N_f}$ depends on the position of the node $N$ and on the edge connecting $N$ to its child $N_f$: the value is indeed computed differently if the edge is a 0-edge or a 1-edge. 

\begin{proposition}
\label{limitelow}
Let  $N$ be a node located at the position $p$ of the node vector. The 0-child of node $N$, $N.\mbox{0-child}$,  has child bound $L_{N.0\mbox{-}child} = p+1$.
\end{proposition}
\begin{proof}
Recall that that the  nodes of the OBDD are inserted in the node vector in the order given by a depth first visit, following first the 0-edges. When the node $N$ is inserted two cases may occur:
\begin{enumerate}
\item $N.\mbox{0-child}$ has already been inserted in the vector;
\item $N.\mbox{0-child}$ has not yet been inserted in the vector.
\end{enumerate}
In the first case, $N.\mbox{0-child}$ occupies a position $p_{0} < p$. In the second case, since we are visiting the OBDD in depth first order, with priority on the 0-edges, $N.\mbox{0-child}$ will be inserted immediately after $N$, and its location will then be 
  $p+1$.
\hfill\end{proof}

\begin{example}
Consider  the OBDD $B$ in Figure~\ref{fig:exbddegdes}, and its node vector $\mathcal{V}_B$. Suppose that the 0-edge of node $b$ is corrupted. As $b$ is located at the position $p_b = 1$ of $\mathcal{V}_B$, the child bound  of its 0-child is $L_{N.0\mbox{-}child} = p_b+1 = 2$. Observe that the 0-child of $b$, $c$, occupies precisely the position 2.

Now, suppose that the 0-edge of node $f$ is corrupted. In this case we have $L_{N.0\mbox{-}child} = p_f+1 = 8$, as the position of $f$ in  $\mathcal{V}_B$ is $p_f=7$. The 0-child of $f$, $d$, was inserted in $\mathcal{V}_B$ before $f$, and thus occupies a position 
$p_d<L_{N.0\mbox{-}child}$. 
\end{example}

For the 1-edges, the child bound is computed in a different way, as these edges are visited only when the subgraph rooted in the 0-child of the current node has been visited and inserted in the node vector.
\begin{proposition}
\label{limitehigh}
Let  $N$ be a node located at the position $p$ of the node vector, let $N.\mbox{0-child}$ and $N.\mbox{1-child}$ be its children, and let  $B_{0}$ be the subgraph rooted in $N.\mbox{0-child}$. Then, the child bound of $N.\mbox{1-child}$ is $L_{N.1\mbox{-}child} = p+|B_{0}|+1$.
\end{proposition}
\begin{proof}
As before, when the node $N$ is inserted two cases may occur:
\begin{enumerate}
\item $N.\mbox{1-child}$ has already been inserted in the vector;
\item $N.\mbox{1-child}$ has not yet been inserted in the vector.
\end{enumerate}
In the first case, $N.\mbox{1-child}$ occupies a position $p_{1} < p$. In the second case, since we are visiting the OBDD in depth first order  with priority on the 0-edges, $N.\mbox{1-child}$ will be inserted only after the visit of the subgraph $B_{0}$ has been completed, and all its nodes have been inserted in the node vector. As some  nodes of $B_{0}$  might be already stored in the node vector when $N$ is visited, the position of $N.\mbox{1-child}$ will be less or equal to $p+|B_{0}|+1$.
\hfill\end{proof}

\begin{example}
Consider  the OBDD $B$ in Figure~\ref{fig:exbddegdes}, and its node vector $\mathcal{V}_B$. Suppose that the 1-edge of node $b$ is corrupted. As the subgraph rooted in the 0-child $c$ of $b$ has dimension 5, and $b$ is located at the position $p_b = 1$ of $\mathcal{V}_B$, the child bound  of its 1-child is $L_{N.1\mbox{-}child} = p_b + 5 + 1 = 7$.
\end{example}

We can now define the set of nodes that could possibly be the children of a corrupted node. 
\begin{definition} 
Let $N$ be a node on level $l_N$, whose pointer to the child $N_f$ is corrupted.
The set of nodes that could be $N_f$ is the set  $$\mathcal{S}_{N_f} = \{ N_i \in \mathcal{V} \ |\ 0 \leq i \leq L_{N_f}  \ \wedge \   l_{N_i} > l_N\}\,.$$
\end{definition}
Observe that the nodes on levels less or equal to $l_N$ have been removed from the set, as  they cannot be children of $N$.
Let us now prove that this set certainly contains $N_f$.

\begin{proposition}
Let $N$ be a node whose pointer to the child $N_f$ is corrupted. Then, $N_f \in \mathcal{S}_{N_f}$.
\end{proposition}
\begin{proof}
By contradiction, suppose that  $N_f \not\in \mathcal{S}_{N_f}$. Thus, either the position of $N_f$ in the node vector is not included in the range $[0,L_{N_f}]$, or $l_{N_f} \le l_N$.
Observe that the first case cannot occur as the position of $N_f$  cannot be negative, and must be less or equal to the child bound $L_{N_f}$, as proved in Propositions~\ref{limitelow} and~\ref{limitehigh}. 
Finally, also the second case cannot occur, since if $N$ is the parent of  $N_f$, $l_N$ must be strictly less than $l_{N_f}$.
\hfill\end{proof}\\
We can now define the reconstruction algorithm for  corrupted pointers.

Observe that the node vector contains completely correct data and this guarantees the correctness of the algorithm.

First of all, the algorithm, through  the function \emph{ComputeLimit}, computes the child bound for the corrupted pointer, applying Propositions~\ref{limitelow} and~\ref{limitehigh}. 
\emph{ComputeLimit} uses the  functions \emph{PositionOf} and \emph{NodeOf} for computing the position of $N$ in the node vector and the dimension of the subgraph rooted in $N_{0}$, respectively.
Then, the set $S_{N_f}$ of all nodes that could be  children of $N$ is defined and examined. For each node $N_f$  in $S_{N_f}$, the algorithm  verifies the nodes corresponding to Hash($N$.0-child, $N_f$), if the corrupted pointer is a 1-edge, 
or to Hash($N_f$, $N$.1-child), if the corrupted pointer is a 0-edge; if one of them is the pointer to $N$, $N_f$ is the correct pointer.

Under certain circumstances, this strategy returns a wrong value. This is due to collisions on the hash tables. Since the values of the pointers are stored as a hash digest, it is not possible to calculate their original values from the unique table. In particular, let $N_1$ and $N_2$ be two pointers to be checked where $N_2$ is the correct one, and let $N$.1-child be the healthy edge of the node. A wrong reconstruction can happen when $H(N_1, N.\mbox{1-child}) = H(N_2, N.\mbox{1-child})$ and $N_1$ is checked before $N_2$. This problem can be handled by using a perfect hashing function, in order to avoid collisions, or using a good hash function that reduces the number of collisions as shown in the Experimental Result Section.

\begin{algorithm}[Reconstruction of the faulty edges]
\label{fig:algopuntatori} 
\ \ 

\begin{scriptsize}
\hrule 
\begin{tabbing}
{\bf INPUT}\\
$N$ /* Address of the node with a corrupted pointer */\\
$BDD$ /* OBDD containing $N$ */\\
$\mathcal{V}$ /* Node vector of the OBDD */\\
$EdgeType$ /* Variable indicating whether the pointer is a 1-edge or a 0-edge */\\
{\bf OUTPUT}\\
\var{Pointer} /* Correct pointer */\\ \\

\var{$L_{N_f}$} = {\em ComputeLimit}(N, EdgeType, $\mathcal{V}$) \\
\var{$\mathcal{S}_{N_f}$} = $\{ N_i \in \mathcal{V} \ |\ 0 \leq i \leq L_{N_f} \ \wedge \ l_{N_i} > l_N\}$ \\
\com{for} \= \com{each} $N_f\in \mathcal{S}_{N_f}$ \com{do} \\
\uno		/* unique subtable for the index $l_N$ */ \\
\uno		uniqueTable \= = OBDD.UniqueTables[$l_N$] \\
\uno		\com{if}\= (EdgeType == 1) \\
\due 		node\= = uniqueTable[Hash($N$.0-child, $N_f$)] \\
\uno		\com{else} \\
\due 		node\= = uniqueTable[Hash($N_f$, $N$.1-child)] \\
\uno			\com{while} \= (node $\not= N$ $\land$ node.Next $\not=$ NULL)\\
\due				/* visit of the collision list */ \\
\due				node = node.Next \\
\uno			\com{if}(node == N)\\
\due		\com{return} $N_f$ \\

\end{tabbing}
\vspace{-0.4cm}   {\em ComputeLimit}(N, EdgeType, $\mathcal{V}$) \rm \vspace{-8pt}
\begin{tabbing}
$p_N$ = \emph{PositionOf}($N$, $\mathcal{V}$) \\
\com{if}\=(EdgeType == 0) \\
\uno \emph{return} $p_N$ + 1 \\
\com{else} \\
\uno\com{return} $p_N$ + \emph{NodesOf}($N$.0-child) + 1
\end{tabbing}
\hrule
\end{scriptsize}
\end{algorithm}

Different strategies could be adopted to handle the possibility of failed recoveries. 
A conservative strategy could consist in reporting the fault as a fatal error any time the involved collision list of the unique subtable contains more than one pointer. Otherwise, we could slightly modify the content of the unique table by adding redundant information, to be exploited for the correct recovery of a pointer,  in case of collisions.
More in general, it would be probably convenient to adopt and apply to the particular structure of OBDDs,  the same strategies developed for  designing error resilient pointer based data structures~\cite{A96}.

\section{Experimental Results}
\label{exp}

We have tested our methods on the classical benchmarks taken from LGSynth93~\cite{Y91}. These benchmarks are relevant especially for logic synthesis applications, where OBDDs are widely applied. Each output has been  separately considered and the OBDDs have been constructed using the ON-set and DC-set of the benchmarks.

The first set of experiments has the purpose of computing the size of the index-resilient reduced OBDDs derived with Algorithm~\ref{fig:algoriduzioneNEW} in order to verify the memory gain of the proposed model with respect to the quasi-reduced OBDD model (QR-OBDDs). 

In order to evaluate the practical memory requirement (number of nodes) of the OBDDs generated by Algorithm~\ref{fig:algoriduzioneNEW}, we have implemented it in C and generated the quasi-reduced OBDD (QR-OBDD), the reduced OBDD (ROBDD) and the index-resilient OBDD (IR-OBDD)  for each considered benchmark. For the sake of briefness, we report in Table~\ref{Tab:index} only a significant subset of the results. The first column reports the name of the instance considered. The following two ones provide its input and output size. Then, the last three columns report the number of internal nodes for QR-OBDD, ROBDD and IR-OBDD considering each output separately. 

\begin{table}[t]
  \centering
  \caption{Number of internal nodes for quasi-reduced, reduced and index-resilient OBDDs.}
\begin{scriptsize}
    \begin{tabular}{|l|rr|rrr|}
       \hline
   {\bf Benchmark}  &  {\bf in}   &  {\bf out}    &  {\bf QR-OBDD }  &  {\bf ROBDD}   & {\bf IR-OBDD} \\
   \hline
\hline
 	al2   & 16    & 47    & 1218  & 269   & 504 \\
    alcom & 15    & 38    & 946   & 175   & 424 \\
    alu1  & 12    & 8     & 206   & 31    & 109 \\
    amd   & 14    & 24    & 1318  & 739   & 1021 \\
    b10   & 15    & 11    & 985   & 617   & 815 \\
    b2    & 16    & 17    & 6613  & 5568  & 5902 \\
    b9    & 16    & 5     & 453   & 196   & 334 \\
    br1   & 12    & 8     & 346   & 242   & 265 \\
    br2   & 12    & 8     & 285   & 174   & 190 \\
    clpl  & 11    & 5     & 140   & 53    & 84 \\
    co14  & 14    & 1     & 39    & 27    & 27 \\
    gary  & 15    & 11    & 988   & 625   & 814 \\
    in2   & 19    & 10    & 4006  & 2476  & 2988 \\
    intb  & 15    & 7     & 1862  & 1228  & 1631 \\
    mp2d  & 14    & 14    & 413   & 151   & 299 \\
    newapla & 12    & 10    & 272   & 78    & 134 \\
    newapla1 & 12    & 7     & 155   & 50    & 81 \\
    newtpla & 15    & 5     & 186   & 83    & 120 \\
    opa   & 17    & 69    & 3091  & 1164  & 2315 \\
    pdc   & 16    & 40    & 6204  & 4754  & 5563 \\
    ryy6  & 16    & 1     & 50    & 23    & 32 \\
    shift & 19    & 10    & 1206  & 189   & 667 \\
    t2    & 17    & 16    & 728   & 306   & 434 \\
    t3    & 12    & 8     & 300   & 111   & 227 \\
    t4    & 12    & 8     & 399   & 213   & 320 \\
    test2 & 11    & 35    & 11678 & 11195 & 11431 \\
    tial  & 14    & 8     & 2230  & 1677  & 1934 \\
   \hline
    \end{tabular}%
  \label{Tab:index}%
\end{scriptsize}
\end{table}%

The results show that IR-OBDDs are an interesting trade-off between memory requirements and error resilience.  
In particular, our algorithm for index-resilient OBDDs nearly always improves the size of the starting quasi-reduced OBDD.   In fact, starting from a QR-OBDD, our new reduction algorithm produces an IR-OBDD with an average gain of $17\%$ nodes, while the standard reduction algorithm allows a gain of about $29\%$ nodes (as shown in Table~\ref{Tab:index}).

We have run a second set of experiments to evaluate the frequency of wrong recoveries for faulty pointers in practical data sets~\cite{Y91}.
The algorithm has been implemented in C, using the CUDD library for the representation of the diagrams. CUDD implements the unique table as discussed in Section~\ref{prel},  and implements {\em shared diagrams}. Shared diagrams are used to share common subgraphs between different OBDDs. The benchmarks we used contain multioutput functions. Each output has been represented as a different OBDD and the entire function as a shared diagram. In that way, subgraphs common to two or more outputs are represented only once. The edge reconstruction algorithm is based on the node's children, thus shared diagrams don't affect the results. 
We have measured the ratio between the range to be searched, the actual number of lookups and the number of nodes of the OBDD. The average range to be searched covered 85\% of the nodes, while each reconstruction required, on average, to check the 37\% of the nodes (more details on this experimental evaluation can be found in~\cite{BCL13}).
In order to evaluate the hash table collision impact on the reconstruction correctness, we have run our experiments using three different settings for the starting hash tables dimension (using the CUDD init manager). We have considered the following dimensions: 256 (standard setting for CUDD), 1024, and 2048. The number of correct reconstructions are then:  $91\%$  for the starting dimension 256, $97\%$ for 1024, and $99\%$ for 2048. 
We can therefore conclude that with a reasonable dimension of 2048, we nearly reach full reconstruction.

\section{Conclusion}
\label{concl}

This paper has presented the first systematic study on resilient OBDDs. The paper has exploited redundancies of standard OBDD tools in order to reconstruct faulty information. Moreover, it has proposed a new  canonical model of OBDDs, which guarantees that a node with a faulty index has a  reconstruction cost $O(r)$, where $r$ is the number of nodes with corrupted index.

An interesting new research direction could be a deeper study of error detection in OBDD data structures. Moreover, since some of the proposed strategies do not always allow a complete reconstruction of faulty edges, a possible future work can be the study of different memorization techniques for OBDDs in order to exploit implicit redundancies.

Given the growing interest in data structures based on decision diagrams and their widespread  application in several research fields, it could  be worth studying the resilience of BDDs reduced with rules different from the  classic merge and deletion ones, as for instance the {\em zero-suppressed decision diagrams} (ZDDs),   widely used in  data mining~\cite{M93,Minato10,Minato13}.

\bibliographystyle{IEEEtranS}
\bibliography{spp}

\appendix
\section*{Appendix: The Apply Algorithm}
\label{opBDD}
\noindent
In this appendix we review the  algorithm  used for implementing the main operations on OBDDs.
For a more comprehensive treatment,  see~\cite{B86, B92}.

\medskip

All the binary Boolean operators on OBDDs are implemented by a general algorithm called {\sc Apply}. This algorithm takes in input a binary Boolean operator {\bf op} together with  two OBDDs  $B_f$ and $B_g$ with the same variable ordering,  representing two  functions $f$ and $g$, and computes the OBDD representing the function $f \mbox{ {\bf op} } g$  defined as
$$(f \mbox{ {\bf op} } g)\, (x_1, \ldots, x_n) = f (x_1, \ldots, x_n) \mbox{ {\bf op} } g (x_1, \ldots, x_n)\,.
$$
The  resulting  OBDD  obeys the same variable ordering of $B_f$ and $B_g$.

The algorithm {\sc Apply} thus provides a basic method for constructing the OBDD representation of any Boolean function $f$ starting from a  Boolean expression or logic gate network representing $f$ with a given set of binary operators.  

The	implementation	of	{\sc Apply}  relies on  the Shannon expansion 
$$f = \overline x_i  f|_{\overline x_i} + x_i f|_{x_i}\,,$$ 
where $f|_{\overline x_i}$ and $f|_{x_i}$ are the restrictions, or cofactors, of the function $f$ obtained assigning the constant values 0 and 1 to the input variable $x_i$, respectively. In particular, 	{\sc Apply}  exploits the fact that the algebraic operations commute with the Shannon expansion for any variable $x_i$, meaning that 
$$
(f \mbox{ {\bf op} } g) = \overline x_i  (f|_{\overline x_i} \mbox{ {\bf op} } g|_{\overline x_i})  + x_i (f|_{x_i} \mbox{ {\bf op} } g|_{x_i})\,.
$$
Thus, we can compute recursively the OBDD $B_{\,\mbox{\bf op}}$ representing $(f \mbox{ {\bf op} } g)$: we start from the root of $B_f$ and $B_g$ and we construct $B_{\,\mbox{\bf op}}$  by recursively constructing the OBDDs representing $(f|_{\overline x_0} \mbox{ {\bf op} } g|_{\overline x_0})$ and $(f|_{x_0} \mbox{ {\bf op} } g|_{x_0})$, where $x_0$ is the first input variable in the common variable ordering of $B_f$ and $B_g$; the roots of these two OBDDs represent respectively the $0$-child  and the $1$-child of the root of $B_{\,\mbox{\bf op}}$, labeled by $x_0$.

More precisely, suppose to execute {\sc Apply} on two OBDDs $B_f$ and $B_g$ with roots $N$ and $M$, respectively. We must consider several cases. If $N$ and $M$ are terminal nodes, a new terminal node is computed having the value of {\bf op} applied to the two constants labeling $N$ and $M$. Otherwise, if at least one node is non-terminal we proceed according to the index of the nodes: 
\begin{itemize}
\item  If the two nodes have the same index  $i$, we create a new node $U$ with index $i$, and we apply the algorithm recursively on $N.0$-child and  $M.0$-child to generate the OBDD whose root becomes the $0$-child of $U$,  and on $N.1$-child and  $M.1$-child to generate the OBDD whose root becomes the $1$-child of $U$.
\item If they have different indices, we proceed by pairing the $0$ and $1$-child of the node with lowest index with the  other node. Suppose for instance that $N.index = i$, but either $M$ is a terminal node, or $M.index > i$. This means that the function $g$ represented by the OBDD with root $M$ does not depend on $x_i$, i.e., $g|_{\overline x_i} = g|_{x_i} = g$, and therefore $(f \mbox{ {\bf op} } g) = \overline x_i  (f|_{\overline x_i} \mbox{ {\bf op} } g)  + x_i (f|_{x_i} \mbox{ {\bf op} } g)$. Hence, we create a new node $U$ with index $i$, and recursively apply the algorithm on $N.0$-child and $M$ to generate the OBDD whose root becomes the $0$-child of $U$, and on $N.1$-child and $M$ to generate the OBDD whose root becomes the $1$-child of $U$. A similar procedure is used in the reverse case, where $M.index < N.index$.
\end{itemize}
To implement the {\sc Apply} algorithm efficiently, two refinements are added. The first one  is used to avoid an exponential blow-up of the recursive calls and consists in maintaining a table $M_A$ of results of the form $M_A[N,M]=U$, indicating that the result of applying the algorithm to the OBDDs with roots $N$ and $M$ is the OBDD with root $U$.  Then, before executing {\sc Apply} on a pair of nodes, we first check whether the table contains an entry for these two nodes. If so, the results can be immediately returned without any further computation. Otherwise, we compute the result of {\sc Apply} on $N$ and $M$, and add a new entry on the table $M_A$ before returning the result.
The second refinement is based on the use of the unique table to ensure that the OBDD computed by the algorithm {\sc Apply} is reduced, i.e., it does not contain  isomorphic subgraphs and redundant nodes.

If  the table $M_A$  is implemented with constant look-up and insertion time (e.g., as a two-dimensional array or as a dynamic hash table with a perfect hashing function producing no collisions), the complexity of the {\sc Apply} procedure is $O(|B_f| |B_g|)$.

\end{document}